\documentclass[amsmath,amssymb,showpacs,superscriptaddress,nofootinbib]{revtex4-1}

\usepackage{bm,amsmath,amsthm,amssymb,hyperref,graphicx,color}
\usepackage[toc,page]{appendix}

\newtheorem{thm}{Theorem}[section]
\newtheorem{lem}{Lemma}[section]

\newcommand{\be}{\begin{equation}}
\newcommand{\ee}{\end{equation}}
\newcommand{\bea}{\begin{eqnarray}}
\newcommand{\eea}{\end{eqnarray}}
\newcommand{\6}{\partial}

\newcommand{\inti}{\int_{-\infty}^{+\infty}}

\newcommand{\lam}{\lambda}

\newcommand{\la}{v}
\newcommand{\m}{w}

\newcommand{\g}{\gamma}

\newcommand{\au}{\mathfrak{a}_1}
\newcommand{\ad}{\mathfrak{a}_2}

\newcommand{\pu}{\varphi_1}
\newcommand{\pd}{\varphi_2}
\newcommand{\pt}{\varphi_3}

\newcommand{\R}{\textsf{R}}
\newcommand{\Lo}{\textsf{L}}
\newcommand{\T}{\textsf{T}}
\newcommand{\Pe}{\textsf{P}}
\newcommand{\tr}{\textsf{t}}

\newcommand{\lu}{v^{(1)}}
\newcommand{\ld}{v^{(2)}}

\begin{document}

\title{Thermodynamics, density profiles and correlation functions of the inhomogeneous one-dimensional spinor Bose gas}
\author{Ovidiu I. P\^a\c tu  }
\affiliation{Institute for Space Sciences, Bucharest-M\u{a}gurele, R 077125, Romania}
\author{Andreas Kl\"umper}
\affiliation{Fachbereich C – Physik, Bergische Universit\"at  Wuppertal, 42097 Wuppertal, Germany}

\pacs{67.85.-d, 02.30.Ik}

\begin{abstract}

We investigate the finite temperature properties of the one-dimensional
two-component Bose gas (2CBG) with repulsive contact interaction in a harmonic
trap. Making use of a new lattice embedding for the 2CBG and the quantum
transfer matrix we derive a system of two nonlinear integral equations
characterizing the thermodynamics of the uniform system for all values of the
relevant parameters: temperature, strength of the interaction, chemical
potential and magnetic field. This system allows for an easy numerical
implementation in stark contrast with the infinite number of equations
obtained by employing the thermodynamic Bethe ansatz. We use this exact
solution coupled with the local density approximation to compute the density
profiles and local density correlation function of the inhomogeneous gas for a
wide range of coupling strengths and temperatures. Our results show that the
polarization in the center of the trap influences heavily the local correlator
especially in the experimentally accessible Tonks-Girardeau regime.

\end{abstract}

\maketitle

\section{Introduction}

As a result of the growing expertise in the creation and the manipulation of
ultracold atomic gases we have witnessed in recent years the experimental
realization of various strongly interacting many-body systems characterized by
an extreme degree of purity and excellent control of relevant parameters
including temperature, strength of the interaction and even dimensionality
\cite{BDZ,CCGOR, GBL}. In particular, one-dimensional (1D) gases which can be
created using optical lattices or atom chips are extremely interesting for
several reasons. One reason is that these systems exhibit regimes such as the
Luttinger liquid (LL) \cite{Gia}, spin-incoherent LL \cite{CZ1,FB,Fiete} and
the ferromagnetic liquid \cite{AT,MF,KG,ZCG3}, which are not present in 2D and
3D. In addition, the high degree of control of the relevant parameters means
that some of these systems can be well approximated by integrable systems
providing a parameter-free comparison of theoretical predictions with
measurements. The paradigmatic example in this case is the experimental
realization \cite{tMSK,tKWW1,tP,tLOH,tPRD,tAetal} of the Lieb-Liniger model
\cite{LL} which has been investigated theoretically for more than fifty years.

It is natural to expect that one-dimensional physical systems comprised of
particles with internal degrees of freedom present a wider range of phenomena
than their single component counterparts. Such spinor gases have been produced
\cite{MLHNW,eSKBS, WMGKH,HSICL,SHLVS,A1etal,LRPLP} by trapping atoms in two or
more internal states which can be referred as (pseudo)spin states. In the case
of fermionic spinor gases the groundstate is antiferromagnetic \cite{LiMa} and
at low-energy they are described by the LL theory. In addition the phenomenon
of spin-charge separation is present which means that the Hilbert space
separates in two independent sectors one for the collective spin excitations
and one for the collective charge excitations. The situation involving bosonic
spinor gases is more complex. While there are cases in which the physics is
similar to their fermionic counterpart \cite{KKMGS}, in the case of
spin-independent interactions, which is integrable and will be the main
subject of this paper, the groundstate is ferromagnetic \cite{EL,KL}. The spin
excitations have a dispersion behaving like $k^2$ which makes the application
of the LL theory impossible. For these systems the low-energy sector is
described by a new universality class called the ferromagnetic liquid
\cite{AT,MF,KG,ZCG3}.

In this article we are going to investigate the thermodynamics, density
profiles and local density-density correlation function of the trapped 1D
two-component Bose gas interacting via a $\delta$-function potential using the
exact solution of the uniform system and the local density approximation
(LDA). From the historical point of view, the first exact solution of an
integrable multicomponent system was obtained by Yang \cite{Y1} and Gaudin
\cite{Ga1} for the two-component fermionic case using what we call nowadays
the nested Bethe ansatz. The spectrum of the 2CBG was derived in
\cite{B1,LGYE} and the low-lying excitations were investigated in
\cite{LGYE}. The spin-wave excitations, spin dynamics, edge exponent in the
dynamic spin structure factor, groundstate properties in a trap at $T=0$ and
impurity dynamics can be found in \cite{FGKS,BBGO,ZCG1,ZCG2,
  HC1,HC2,ZMS,HZGC,KMGZD,MZD,RCK}. Even though the study of nonlocal
correlation functions is still in its early stages \cite{IP,CSZ2} recent
progress in the calculation of form-factors \cite{POK,BPRS,PRS} opens the way
for the calculation of correlators along the lines of \cite{CCS,PC}. As a
result of the integrability of the model one might believe that computing the
thermodynamical properties should be an accessible task. Unfortunately, this
assertion is not true. For example, the application of the thermodynamical
Bethe ansatz (TBA) \cite{YY,T} produces an infinite number of nonlinear
integral equations (NLIEs) \cite{GLYZ} which makes the extraction of physical
information extremely difficult. Even though some important results can be
derived on the basis of the TBA equations \cite{GBT,CKB,KC}, it is highly
desirable to obtain an alternative thermodynamic description allowing for an
easy numerical implementation. Our result, Eqs.~(\ref{GCPgas}) and
(\ref{NLIEgas}), which was derived using the quantum transfer matrix (QTM) and
a new lattice embedding of the 2CBG, provide such an alternative description.

The basis of our method is the following observation. Consider an integrable
continuum model and one of its lattice embeddings. By lattice embedding of a
continuum model we understand an integrable model defined on a lattice whose
Bethe equations and spectrum transform in a specific scaling limit in the
Bethe equations and spectrum of the continuum model. Then, the thermodynamics
of the continuum model can be derived from the thermodynamics of the lattice
model if we perform the same scaling limit as in the case of the Bethe
equations and spectrum. One example is the Yang-Yang thermodynamics \cite{YY}
of the Lieb-Liniger model which was obtained from the thermodynamics of the
critical XXZ spin chain \cite{SBGA}.  As we will show in Sect. \ref{S2}, in
the case of the 2CBG the relevant lattice model is the critical $q=3$
Perk-Schultz spin chain \cite{PS1,Schul1,BVV,dV,dVL,L}. The thermodynamics of
the spin chain is derived using the quantum transfer matrix method
\cite{MS,Koma,SI1,K1,K2}, which can be defined only for lattice models but,
has the fundamental advantage of providing a finite number of integral
equations. Performing the scaling limit in these equations we obtain
Eqs.~(\ref{GCPgas}) and (\ref{NLIEgas}). This finite system of equations
coupled with the LDA allows for the calculation of the density profiles and
the local density-density correlation function of the system subjected to a
slowly varying harmonic potential.

The plan of the paper is as follows. In Section \ref{S1} we introduce the 2CBG
and the thermodynamic description of the uniform system. The results for the
density profiles and local correlator are presented in Section \ref{SD}.
Section \ref{S2} contains the lattice embedding of the 2CBG and in Section
\ref{S3} we remind some notions of the algebraic Bethe ansatz and introduce
the quantum transfer matrix. The derivation of the free energy of the spin
chain can be found in Sect.  \ref{S4} and the continuum limit is performed in
Section \ref{S5}. The nested Bethe ansatz solution of the generalized $q=3$
Perk-Schultz model and the proof of several useful identities can be found in
Appendices \ref{ap11} and \ref{ap22}. Some of the results presented
in this article were announced in the short paper \cite{KP1}.


\section{The one-dimensional two-component  Bose gas  }\label{S1}

We consider a one-dimensional system of equal mass bosons with two internal
degrees of freedom interacting via a $\delta$-function potential and
constrained on a ring of circumference $L_B$ (periodic boundary
conditions). The second quantized Hamiltonian in the presence of a harmonic
trapping potential $V(x)=m\omega x^2/2$ is
\begin{align}\label{Hc}
\mathcal{H}_{Bose}=\int_0^{L_B} dx \left[\frac{\hbar^2}{2m} (\6_x \Psi^\dagger \6_x \Psi)
+\frac{g}{2}\, :(\Psi^\dagger \Psi)^2:+(V(x)-\mu)(\Psi^\dagger \Psi)-H(\Psi^\dagger\sigma_z\Psi)\right]\, ,
\end{align}
with
$
\Psi=\left(\begin{array}{c} \Psi_1(x)\\ \Psi_2(x)\end{array}\right)\, ,
\Psi^\dagger=\left(\begin{array}{cc} \Psi_1^\dagger(x) & \Psi_2^\dagger(x)\end{array}\right)\, ,
\sigma_z=\left(\begin{array}{cc}1&0\\ 0&-1\end{array}\right)\, ,
$
where $\Psi_a(x),\, a=\{1,2\}$ are 1D quantum fields satisfying  canonical commutation
relations
$
\Psi_a(x)\Psi_b^\dagger(y)-\Psi_b^\dagger(y)\Psi_a(x)=\delta_{ab}\delta(x-y)\, .
$
In (\ref{Hc}) $g>0$ is the coupling constant, $\mu$ is the chemical potential,
$H$ the external magnetic field (the Bohr magneton and the Lande factors are
absorbed into $H$), $\omega$ is the trap oscillation frequency, $m$ is the
mass of the particles and $:\ \ :$ denotes normal ordering. In experiments,
the internal degrees of freedom are two distinguishable hyperfine states which
can be thought of as a (pseudo)spin $\frac{1}{2}$.  Compared with the scalar
case (Lieb-Liniger model \cite{LL}) the wavefunctions of the 2CBG are
symmetric only under exchange of coordinates of particles with the same
spin.
One dimensional systems that are well approximated by the
Hamiltonian (\ref{Hc}) can be experimentally achieved in highly elongated
cylindrical traps  $(\omega\ll \omega_\bot)$ where $\omega_\bot$ is the trap
oscillation frequency in  the transversal plane. Assuming that the 3D scattering
length $a_{3D}$ is much smaller than the  transverse harmonic oscillator length
$l_\bot=\sqrt{\hbar/m \omega_\bot}$ the  coupling strength can be expressed as
$g\simeq 2\hbar^2 a_{3D}/(m l_\bot^2)=2\hbar\omega_\bot a_{3D}$ \cite{Olsh,BMOl}.
The realization of the 1D regime requires that $\hbar\omega_\bot$ is larger than
the thermal energy $k_BT$ and the chemical potential $\mu$ \cite{KGDS2,BKSl}. The
experimental advances of the last decade paved the way for the experimental realization
of such quasi-1D systems \cite{BDZ,CCGOR,GBL} characterized by coupling strengths
which range from  weak  coupling $(\g\ll 1)$ to the strongly interacting Tonks-Girardeau
regime ($\g\gg 1$). In order to lighten the notation in the following we are going to
consider $\hbar=k_B=2m=1$ and introduce the effective interaction parameter
$c=g m/\hbar^2$.

\subsection{Bethe ansatz solution for the uniform system}

Even though the Hamiltonian (\ref{Hc}) is integrable only when the system is
homogeneous ($V(x)=0$), the trapped system can be efficiently investigated
using the solution of the uniform system coupled with the local density
approximation.  Therefore, we will first review the exact solution of the
homogeneous 2CBG and defer the treatment of the inhomogeneous system to
Section \ref{SD}. Since we are interested in investigating the thermodynamic
behavior we will consider only the case of repulsive interaction $c>0$ (for
$c<0$ the system is thermodynamically unstable). For a system of $M$ particles
of which $M_1^B$ are of type 1 and $M_2^B$ are of type 2 ($M=M_1^B+M_2^B$) the
energy spectrum of the 2CBG obtained using the nested Bethe ansatz
\cite{LGYE,B1} (see also \cite{Kul,KRe1,Slav}) is
\begin{align}\label{eb}
E_{Bose}=\sum_{j=1}^M \bar e_0(k^{(1)}_j)-H(M_1^B-M_2^B)\, ,\ \ \ \ \bar e_0(k)=k^2-\mu\, ,
\end{align}
with $\{k^{(1)}_j\}$ satisfying the Bethe ansatz equations (BAEs)
\begin{subequations}\label{BEc}
\begin{align}
&e^{ik^{(1)}_sL_B}=\prod_{\substack{j=1\\j\ne s}}^M\frac{k^{(1)}_s-k^{(1)}_j+ic}{k^{(1)}_s-k^{(1)}_j-ic}\
\prod_{p=1}^{M_1^B}\frac{k^{(1)}_s-k^{(2)}_p-ic/2}{k^{(1)}_s-k^{(2)}_p+ic/2}\, ,\ \  s=1,\cdots,M\, , \\
&\prod_{j=1}^M\frac{k^{(2)}_l-k^{(1)}_j+ic/2}{k^{(2)}_l-k^{(1)}_j-ic/2}=
\prod_{\substack{p=1\\p\ne l}}^{M_1^B}\frac{k^{(2)}_l-k^{(2)}_p+ic}{k^{(2)}_l-k^{(2)}_p-ic}\, , \ \ l=1,\cdots,M_1^B\, .
\end{align}
\end{subequations}
The solution is characterized by two sets of rapidities,
$\{k^{(1)}_j\},\{k^{(2)}_j\}$ (a general characteristic of integrable
two-component systems) with the second set of rapidities contributing to
the energy (\ref{eb}) only via the BAEs.

Even in the case of integrable models computing the thermodynamics is an
extremely challenging task. One way to tackle this problem is the utilization
of the thermodynamic Bethe ansatz \cite{YY,T}. In this framework, the
grandcanonical potential of the system is ($\beta=1/T$) \cite{GLYZ}
\begin{align}
\phi(\mu,H,\beta)=-\frac{1}{2\pi\beta}\inti dk\,  \ln(1+\eta_1(k))\, ,
\end{align}
with $\eta_1(k)$ satisfying the following infinite system of nonlinear
integral equations
\footnote{ This form of the TBA equations can be easily derived from the
one used in \cite{CKB,KC} by the use of simple transformations and identities
as in  Chap. 12 of \cite{T}).}
\begin{subequations}\label{TBAeq}
\begin{align}
&\ln \eta_1(k)=-\beta(k^2-\mu-H)+a_3\ast f\ast\ln(1+\eta_1(k))+f\ast\ln(1+\eta_2(k))\, ,\\
&\ln \eta_n(k)=f\ast\left(\ln(1+\eta_{n-1}(k))+\ln(1+\eta_{n+1}(k))\right)\, ,\ \ n=2,\cdots,\infty\, ,
\end{align}
\end{subequations}
together with the asymptotic condition $\lim_{n\rightarrow\infty} \ln
\eta_n(k)/n=2\beta H$. In Eqs.~(\ref{TBAeq}) $g*h(k)\equiv \inti
g(k-k')h(k')\ dk'$, $f(k)=1/[2c\cosh(\pi k/c)]$ and
$a_n(k)=nc/[2\pi((nc/2)^2+k^2)]$.  While important, it is obvious that
extracting physically relevant information from this system of equations is
very hard even from the numerical point of view (see \cite{CKB,KC,GBT}),
highlighting the need for a more manageable thermodynamic description of the
2CBG which will be presented in the next section.

\subsection{Alternative thermodynamic description of the uniform 2CBG}

A more efficient thermodynamic description of the 2CBG was proposed in
\cite{KP1}. The full derivation of this result which is based on the
connection of our model with the $(---)$ Perk-Schultz spin chain and the
quantum transfer matrix method represents one of the main results of this
paper and will be presented in Sections \ref{S4} and \ref{S5}. In this
description the grandcanonical potential per unit length of the 2CBG is
\be\label{GCPgas}
\phi(\mu,H,\beta)=-\frac{1}{2\pi\beta}\int_{\mathbb{R}}\left[\ln(1+a_1(k))+\ln(1+a_2(k))\right]\, dk\, ,
\ee
with $a_{1,2}(k)$ auxiliary functions satisfying the following system of integral equations
\begin{subequations}\label{NLIEgas}
\begin{align}
\ln a_1(k)&=-\beta(k^2-\mu-H)+\int_{\mathbb{R}} K_0^B(k-k')\ln(1+a_1(k'))\, dk'
+\int_{\mathbb{R}} K_2^B(k-k'-i\varepsilon)\ln(1+a_2(k'))\, dk'\, ,\\
\ln a_2(k)&=-\beta(k^2-\mu+H)+\int_{\mathbb{R}} K_1^B(k-k'+i\varepsilon)\ln(1+a_1(k'))\, dk'
+\int_{\mathbb{R}} K_0^B(k-k')\ln(1+a_2(k'))\, dk'\, ,
\end{align}
\end{subequations}
where $\varepsilon\rightarrow 0$ and
\be
 K_0^B(k)=\frac{1}{2\pi} \frac{2c}{k^2+c^2}\, ,\ \
 K_1^B(k)=\frac{1}{2\pi }\frac{c}{k(k+ic)}\, ,\ \
 K_2^B(k)=\frac{1}{2\pi }\frac{c}{k(k-ic)}\, .
\ee

We can check the validity of our result in three well known limits. First, we
will address the noninteracting limit, $c\rightarrow 0$.  Using
$\lim_{c\rightarrow 0} K_2^B(k-i\varepsilon)$ $=\lim_{c\rightarrow 0}
K_1^B(k+i\varepsilon)=0$ and $\lim_{c\rightarrow 0} K_0^B(k-k')=\delta(k-k')$
the NLIEs (\ref{NLIEgas}) decouple with the result
\begin{align*}
\ln a_1(k)&=-\beta(k^2-\mu-H)+\ln(1+a_1(k))\, ,\\
\ln a_2(k)&=-\beta(k^2-\mu+H)+\ln(1+a_2(k))\, .
\end{align*}
These equations are easily solved obtaining
\[
\phi(\mu,H,\beta)=\frac{1}{2\pi\beta}\int_{\mathbb{R}}
\left[\ln(1-e^{-\beta(k^2-\mu-H)})+\ln(1-e^{-\beta(k^2-\mu+H)})\right]dk\, ,
\]
which is exactly the grandcanonical potential of two noninteracting Bose gases
at different chemical potentials.

In the strong magnetic field limit, $H\rightarrow \infty$, due to the fact
that the inhomogeneity $-\beta(k^2-\mu+H)$ is large and negative $a_2(k)\sim
0$, which means that the thermodynamics of the system is given by
\[
\phi(\mu,H,\beta)=-\frac{1}{2\pi\beta}\int_{\mathbb{R}}\ln(1+a_1(k))\,dk\, ,\ \
\log a_1(k)=-\beta(k^2-\mu-H)+\int_{\mathbb{R}} K_0^B(k-k')\ln(1+a_1(k'))\, dk'\,  .
\]
This result, which is the Yang-Yang thermodynamics \cite{YY} of the
Lieb-Liniger gas, confirms the natural expectation that in the strong magnetic
limit the system will become fully polarized, and, therefore, the
thermodynamics will be the same as in the single component case. In a similar
fashion, we can show that if the magnetic field is fixed and finite, then, in
the low temperature limit $(T\ll\mu,H,c)$, the same formulas are obtained proving
that the ground state is ferromagnetic (fully polarized).

In the limit of impenetrable particles, $c\rightarrow \infty$, the system is
effectively ``fermionized" and we should reproduce the result obtained by
Takahashi (Chap. 12 of \cite{T}) for impenetrable repulsive two-component
fermions (2CFG)
\be\label{impfermions}
\phi_F(\mu,H,\beta)=-\frac{1}{2\pi\beta}\int_{\mathbb{R}}dk\,\ln\left(1+2\cosh(\beta H)e^{-\beta(k^2-\mu)}\right)\, .
\ee
Even though we did not succeed in providing an analytic proof of this
equivalence, we have checked numerically and found perfect agreement as can
be seen in Fig.~\ref{Impenetrable}.
\begin{figure}
\includegraphics[width=0.75\linewidth]{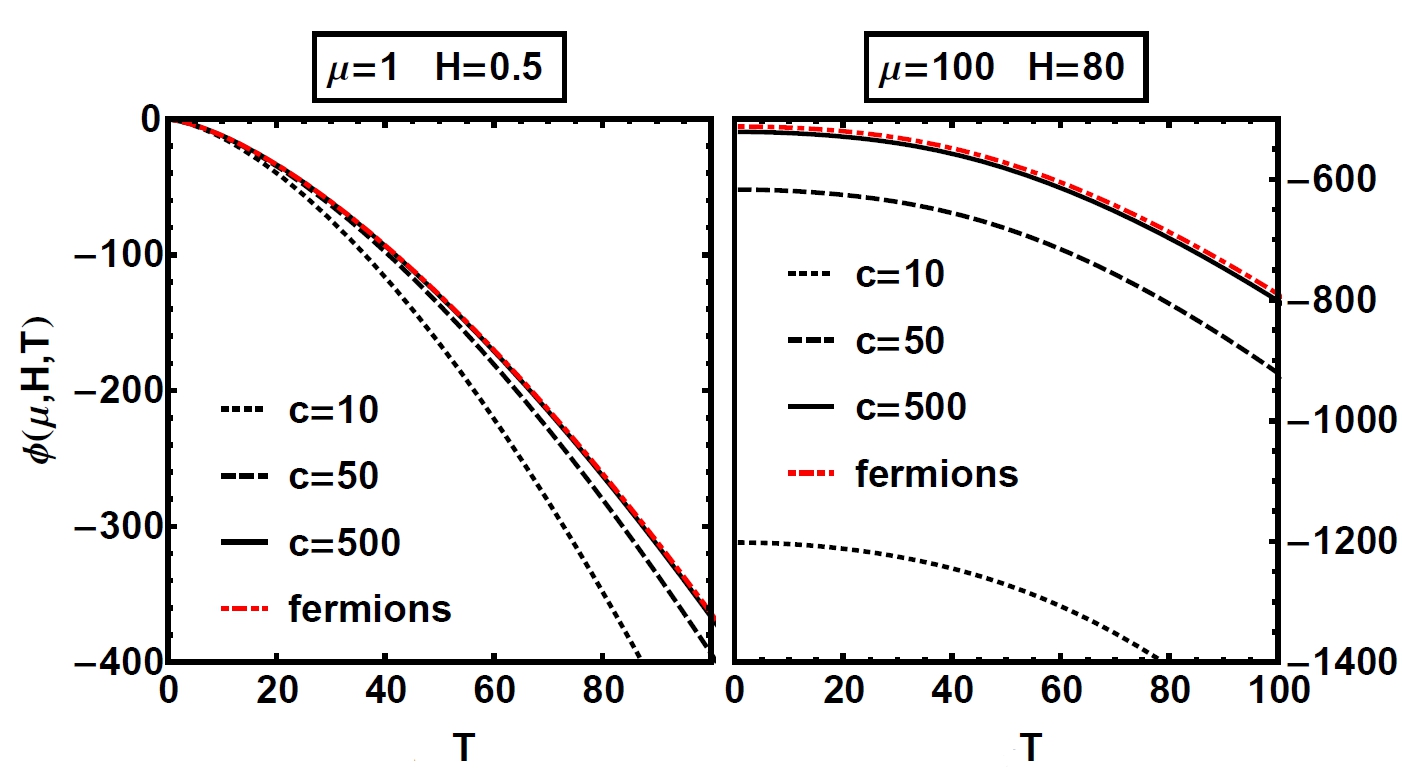}
\caption{(Color online) Plot of the grandcanonical potential per  length of the 2CBG and
  impenetrable 2CFG (dot-dashed line) as a function of temperature for
  $\mu=1\, , H=0.5$ (left panel) and $\mu=100\, , H=80$ (right panel). For large
  values of the coupling constant the 2CBG result approaches asymptotically
  Eq.~(\ref{impfermions}). Grandcanonical potential per length, temperature, chemical potential and magnetic
field in units of $\phi_0$, $T_0$, $\mu_0$ and $h_0$ \cite{units}.
}
\label{Impenetrable}
\end{figure}

\subsection{The ``heuristic'' derivation}

The rigorous and self-contained derivation of the thermodynamic description
presented in the previous section will be given in Sections \ref{S4} and
\ref{S5}. Here we provide a heuristic derivation based on the fact that the
system of equations (\ref{TBAeq}) share the same structure with the NLIEs for
the spin $\frac{1}{2}$ XXX spin chain. Let us be more precise. We consider a
very general $XXX$ spin chain case characterized by the Hamiltonian
$\mathcal{H}_{XXX}= \mathcal{H}_0+\sum_n\beta_n\mathcal{J}_n$ with
nearest-neighbor interaction $\mathcal{H}_0$ and higher conserved currents
$\mathcal{J}_n$. Explicitly we have
$\mathcal{H}_0:=J\sum_i\left(2\overrightarrow{S}_i\overrightarrow{S}_{i+1}+
\frac{1}{2}\right)-\tilde H \sum_i S_i^z\, ,$ with $J$ and $\tilde H$ the
interaction strength and magnetic field. The TBA result for the free energy
per unit length is $ F(\tilde H,\beta)=\tilde e_0-\beta^{-1}\inti dk\, f(k)
\ln(1+\tilde \eta_1(k))\, ,$ ($\tilde e_0$ is the zero point energy and
$f(k)=1/[2\cosh(\pi k)]$, $\tilde f=f+\sum_n\beta_nf^{(n)}$) with
$\tilde\eta_1(k)$ satisfying
\begin{subequations}\label{TBAeqXXX}
\begin{align}
&\ln \tilde \eta_1(k)=-2\pi\beta J \tilde f(k)+f\ast\ln(1+\tilde\eta_2(k))\, ,\\
&\ln \tilde \eta_n(k)=f\ast\left[\ln(1+\tilde\eta_{n-1}(k))+\ln(1+\tilde\eta_{n+1}(k))\right]\, ,\ \ n=2,\cdots,\infty\, ,
\end{align}
\end{subequations}
together with the asymptotic condition $\lim_{n\rightarrow\infty} \ln
\tilde\eta_n(k)/n=\beta \tilde H$. A more compact result can be obtained using the
quantum transfer matrix \cite{K1,K2,K2b} yielding $ F(\tilde H,\beta)=\tilde
e_0-\beta^{-1}\inti dk\, f(k)[\ln(1+\tilde a_1(k))+\ln(1+\tilde
a_2(k))]\, ,$ with only two auxiliary functions satisfying
\begin{subequations}\label{TBAeqXXXQTM}
\begin{align}
&\ln \tilde a_1(k)=-\beta[2\pi J \tilde f(k)-\tilde H/2]+\tilde K_0\ast\ln(1+\tilde a_1(k))-\tilde K_2\ast\ln(1+\tilde a_2(k))\, ,\\
&\ln \tilde a_2(k)=-\beta[2\pi J \tilde f(k)+\tilde H/2]-\tilde K_1\ast\ln(1+\tilde a_1(k))+\tilde K_0\ast\ln(1+\tilde a_2(k))\, .
\end{align}
\end{subequations}
where
\[
\tilde K_0(k)=\frac{1}{2\pi}\inti  \frac{e^{i k x} e^{-|x|/2}}{2\cosh\frac{x}{2}}\, dx\, , \ \
\tilde K_1(k)=\frac{1}{2\pi}\inti  \frac{e^{i k x} e^{x-|x|/2}}{2\cosh\frac{x}{2}}\, dx\, , \ \
\tilde K_2(k)=\frac{1}{2\pi}\inti  \frac{e^{i k x} e^{-x-|x|/2}}{2\cosh\frac{x}{2}}\, dx\, . \ \
\]
Comparing the two systems of equations (\ref{TBAeqXXX}) and
(\ref{TBAeqXXXQTM}) and their corresponding expressions for the free energy we
notice that a) $\ln(1+\tilde\eta_1(k))=\ln(1+\tilde a_1(k))+\ln(1+\tilde
a_2(k))$ and b) the driving terms (modulo the magnetic field) in the QTM
system are the same as the driving term of the integral equation for $\ln
\tilde\eta_1(k)$. Due to the similar form of the TBA equations for the XXX
spin chain (\ref{TBAeqXXX}) -- with practically general function $\tilde f$ --
and the 2CBG (\ref{TBAeq}) it is natural
that a similar system like (\ref{TBAeqXXXQTM}) can be derived via: a)
$k\rightarrow k/c\, $ b) $\ln(1+\eta_1(k))=\ln(1+ a_1(k))+\ln(1+ a_2(k))$ and
c) replacement of the driving terms with $-\beta(k^2-\mu\mp H)+a_3\ast
f\ast\ln(1+\eta_1(k))=$ $-\beta(k^2-\mu\mp H)+a_3\ast f\ast(\ln(1+
a_1(k))+\ln(1+ a_2(k)))$.  Performing these transformations we obtain
(\ref{GCPgas}) for the grandcanonical potential and the following system of
integral equations
\begin{subequations}\label{TBAeqQTMM}
\begin{align}
&\ln  a_1(k)=-\beta(k^2-\mu+ H)+(a_3\ast f+\tilde K_0)\ast\ln(1+\tilde a_1(k))+(a_3\ast f-\tilde K_2)\ast\ln(1+\tilde a_2(k))\, ,\\
&\ln  a_2(k)=-\beta(k^2-\mu- H)+(a_3\ast f-\tilde K_1)\ast\ln(1+\tilde a_1(k))+(a_3\ast f+\tilde K_0)\ast\ln(1+\tilde a_2(k))\, .
\end{align}
\end{subequations}
Finally,   it is easy to
show that $(a_3\ast f+\tilde K_0)(k)=K_0^B(k)\, ,\ $ $(a_3\ast f-\tilde
K_2)(k)=K_1^B(k+i\epsilon)\, ,$ $(a_3\ast f-\tilde K_1)(k)=K_2^B(k-i\epsilon)$
by taking the Fourier transform and applying the convolution theorem,
completing the derivation.


\section{Density profiles and local correlation functions}\label{SD}

The experimentally relevant case of a trapped system ($V(x)\ne 0$ in
(\ref{Hc})) can be investigated using the exact solution of the uniform
system, Eqs.~(\ref{GCPgas}) and (\ref{NLIEgas}), and the local density
approximation. Under this approximation, which is valid for a slowly varying
potential, the system in a trap can be described locally \cite{KGDS2} as a
uniform gas with chemical potential and magnetic field defined by
\be\label{LDA}
\mu(x)=\mu(0)-V(x)\, , \ \ H(x)=H(0)\, ,
\ee
with $\mu(0)$ and $H(0)$ the chemical potential and magnetic field in the
center of the trap.  Therefore, for given values of temperature $T$, coupling
strength $c$, chemical potential and magnetic field in the center of the trap
$\mu(0)$ and $H(0)$, relevant thermodynamic quantities at a distance $x$ from
the center of the trap can be computed using Eqs.~(\ref{GCPgas}) and
(\ref{NLIEgas}) with $\mu$ and $H$ replaced by (\ref{LDA}).

In this section we will be mainly concerned with the calculation of the
density profiles in the trap and the normalized finite temperature local
density-density correlator $g_2^{(T)} $\cite{GSh,KGDS1,KGDS2}. For a uniform
system in thermal equilibrium the linear densities of the two types of bosons
can be obtained from the derivatives of the grandcanonical potential per unit
length (\ref{GCPgas}):
\begin{align}\label{n}
n_1=-\frac{1}{2}\left(\frac{\6 \phi}{\6 \mu}+\frac{\6 \phi}{\6 H}\right)\, ,\ \ \
n_2=-\frac{1}{2}\left(\frac{\6 \phi}{\6 \mu}-\frac{\6 \phi}{\6 H}\right)\, .
\end{align}
In the case of the local correlator, as it was shown in \cite{KGDS2}, a simple
application of the Hellmann-Feynman theorem allows to express this important
observable as the derivative of the grandcanonical potential with respect to
the coupling strength
\be\label{g}
g_2^{(T)}(x)=\frac{\sum_{a,b}\langle\Psi_a^\dagger(x)\Psi_b^\dagger(x)\Psi_a(x)\Psi_b(x)\rangle_T}
{\left(\sum_{a}\langle \Psi_a^\dagger(x)\Psi_a(x)\rangle_T\right)^2}=\frac{1}{(\sum_a n_a)^2}\frac{\6 \phi}{\6 c}\, .
\ee
The local density correlation function of the uniform system which presents a
nonmonotonic behaviour as a function of temperature was investigated by Caux,
Klauser and van den Brink in \cite{CKB,KC}.

A uniform 2CBG can be described with the help of three dimensionless
parameters: interaction parameter $\gamma\equiv c/(n_1+n_2)=c/n$, reduced
temperature $\tau\equiv T/T_d=T/n^2$ and polarization $P\equiv
(n_1-n_2)/(n_1+n_2)$. While this description can be naturally extended to the
inhomogeneous case by replacing the density $n$ with the local value in the
trap $n(x)$ it is preferable to work (see \cite{KGDS2}) with
\be
\gamma(x)=\frac{c}{n(x)}\, ,\ \ t=\frac{T}{c^2}\, ,\ \ P(x)=\frac{n_1(x)-n_2(x)}{n_1(x)+n_2(x)}
\ee
where we have introduced a new temperature parameter $t=\tau(x)/\gamma^2(x)$
which has the advantage of not depending on the local density $n(z)$ and
characterizes the entire system in thermal equilibrium.

\subsection{Regimes in a trapped and uniform 2CBG}

Analytical results on the local density correlation function of the 2CBG at
arbitrary polarization (any value of $H$) are very scarce in the literature (a
notable exception is \cite{KC}). However in the case when the system is fully
polarized $(P=1)$ the 2CBG is equivalent to the Lieb-Liniger model for which a
classification of different regimes exist. The fully polarized case represents
a good starting point for the investigation of the $\gamma-t-P$ parameter space
if we take into account that 1D bosons favour ferromagnetic behavior for spin
independent interactions.
\begin{figure}
\includegraphics[width=0.75\linewidth]{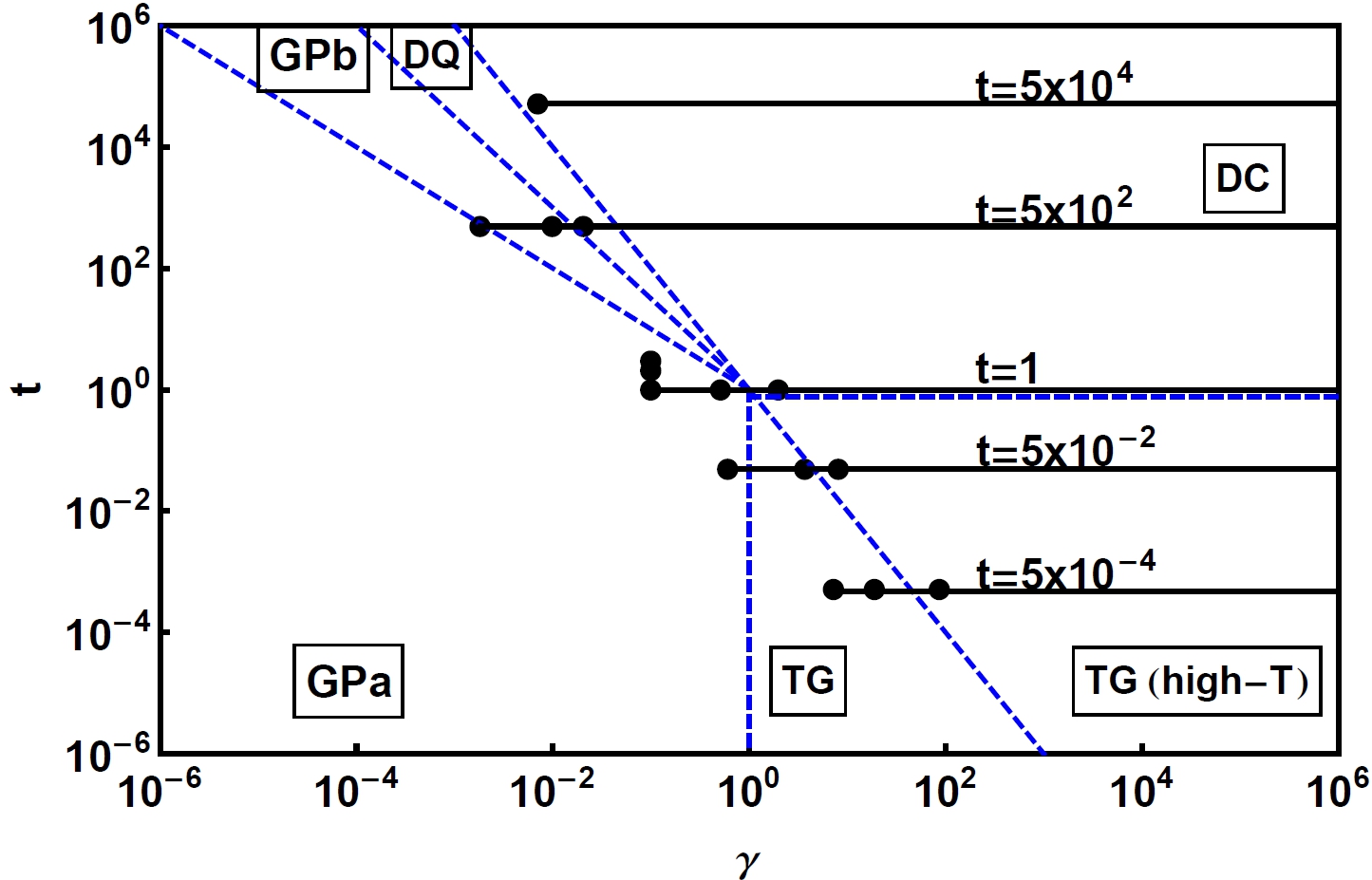}
\caption{(Color online) Diagram of the different regimes of the fully polarized uniform 2CBG (Lieb-Liniger model) in the
$\gamma-t$ plane (boundaries depicted as dashed lines)\cite{KGDS2}. The six regimes depicted are defined as follows:
Tonks-Girardeau (TG) $\gamma\gg 1\, ,t\ll \gamma^{-2}$;
High-Temperature Tonks-Girardeau (TG-high T) $\gamma\gg 1\, ,\gamma^{-2}\ll t\ll 1$;
Gross-Pitaevskii a (GPa) $\gamma\ll 1\, ,t\ll \gamma^{-1}$;
Gross-Pitaevskii b (GPb) $\gamma^{-1}\ll t\ll \gamma^{-3/2}$;
Decoherent Quantum (DQ) $\gamma^{-3/2}\ll t\ll\gamma^{-2}$;
Decoherent Classical $t\gg max(\gamma^{-2},1)$.
The black dots are the values of  $\gamma(0)$ and $t$ in the center of the trap of the representative
systems considered in the next section with the continuous black lines at $t=\{5\times 10^{-4},5\times 10^{-2},1 ,5\times 10^{2},5\times 10^{4}\}$
representing the values of $\gamma(x)$ in the trap for these systems.
}
\label{PD}
\end{figure}
In the case of the Lieb-Liniger gas, Kheruntsyan, Gangardt, Drummond and
Shlyapnikov \cite{KGDS1,KGDS2} identified three regimes (each containing two
sub-regimes) using the properties of the local density correlator as a
function of coupling strength and temperature. Using $\gamma$ and $t$ as
parameters this regimes can be characterized as follows:

{\it Strong coupling regime.} In the limit of strong interaction (or
low-densities), $\gamma\gg 1$, the system becomes equivalent with a system of
free fermions and therefore the local density correlation function is
suppressed, $g_2^{(T)}\ll 1$, as a consequence of the fermionic nature of the
wavefunctions. Two sub-regimes are distinguished: a) Tonks-Girardeau
characterized by $\gamma\gg 1$ and temperature smaller than the degeneracy
temperature $ t\ll \gamma^{-2}$ and b) High-Temperature Tonks-Girardeau for
which $\gamma^{-2}\ll t\ll 1$.

{\it Gross-Pitaevskii regime.} In the limit of vanishing interactions (or
large densities), $\gamma \ll 1$, the local density correlation function
approaches $g_2^{(T)}\simeq 1$ which is the value for free bosons. The two
subregimes are: Gross-Pitaevskii (a) $(\gamma\ll 1 \, ,t\ll \gamma^{-1})$ and
Gross-Pitaevskii (b) $(\gamma^{-1}\ll t\ll \gamma^{-3/2})$.

{\it Decoherent regime.} At high temperatures the local pair correlation is
close to $g_2^{(T)}\simeq 2$.  The two subregimes are: Decoherent Quantum
$(\gamma^{-3/2}\ll t\ll \gamma^{-2})$ and Decoherent Classical $(t\gg
max(\gamma^{-2},1))$.

For a trapped system this classification still remains valid with the
possibility of different regimes coexisting within the trap. Due to the fact
that for large distances from the center of the trap the density vanishes
$n(x)\rightarrow 0, \gamma(x)\rightarrow\infty$ this means that the tails of
the system are always in the TG regime (for $t<1$) or the DC regime (for
$t>1$).  For given values of the interaction and temperature parameters
$\gamma(0)$ and $t$ the density profile follows a straight line at $t$
parallel with the $\gamma$ axis as it is shown in Fig. \ref{PD} for
$t=\{5\times 10^{-4},5\times 10^{-2},1 ,5\times 10^{2},5\times 10^{4}\}$. All
these considerations are valid for the fully polarized 2CBG. In the following
sections we are going to investigate the influence of the polarization on the
density profiles and local density correlation function.

\subsection{Density profiles and local correlation functions at high-temperatures}

\begin{figure}
\includegraphics[width=0.9\linewidth]{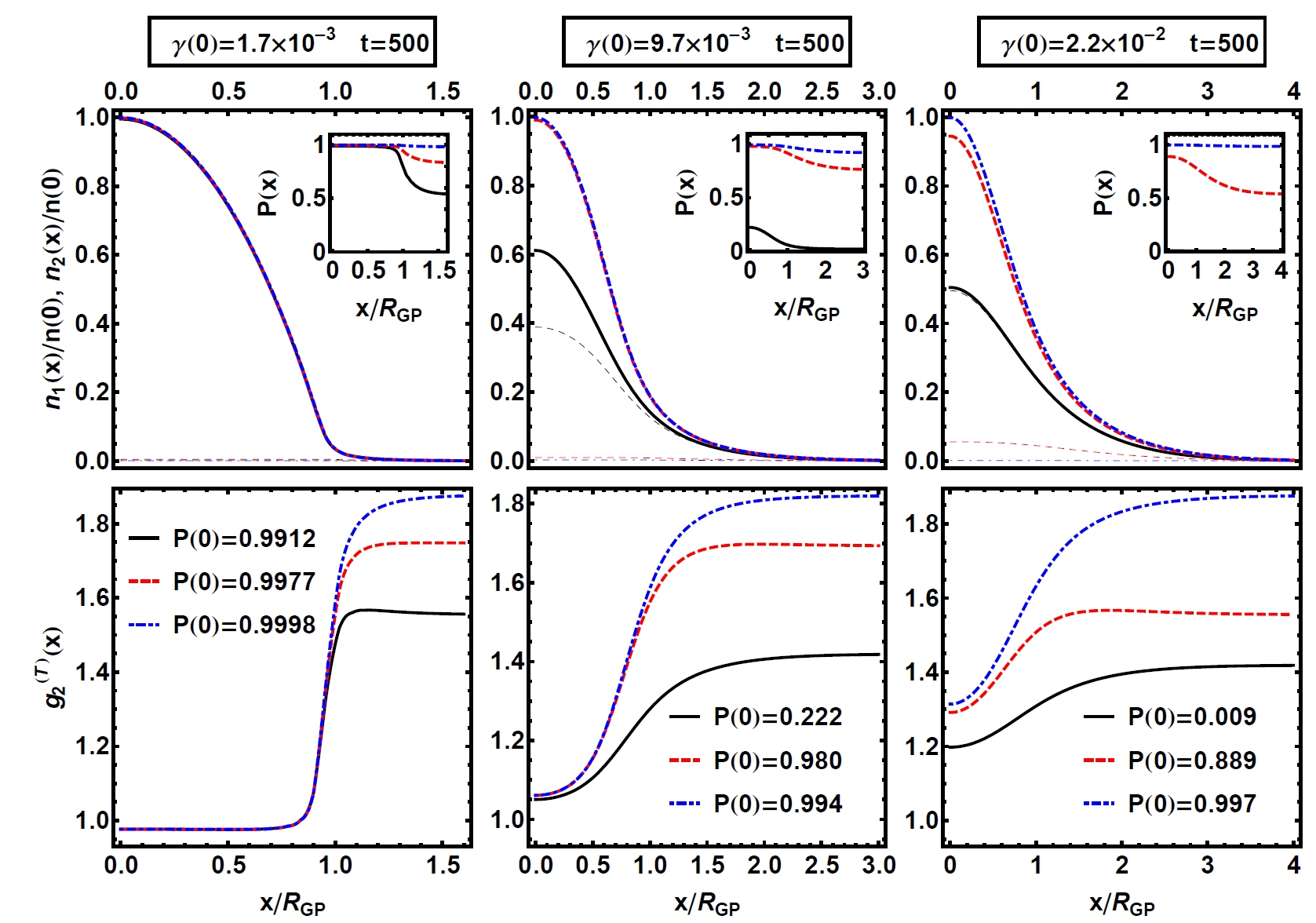}
\caption{(Color online) Density profiles (upper panels) and local density correlation
  function $g_2^{(T)}(x)$ (lower panels) of a trapped 2CBG at temperature
  $t=5\times 10^2$ for three values of $\gamma(0)=\{1.7\times 10^{-3},
  9.7\times 10^{-3},2.2\times 10^{-2}\}$ and various polarizations. In the
  upper panels the thick continuous, thick dashed, and thick dot-dashed (thin continuous,
  thin dashed and thin dot-dashed) lines represent the majority
  (minority) components and $n(0)$ is the total density in the center of the
  trap. The insets depict the variation of the polarization in the trap.  }
\label{T500}
\end{figure}
For a given value of the coupling strength $c$ the region of high temperatures
is defined by $t\gg 1$. We have calculated the density profiles and the local
correlation function on a fine grid of the trap using Eqs.~(\ref{GCPgas}),
(\ref{NLIEgas}), (\ref{n}), (\ref{g}) and the LDA chemical potential and
magnetic field (\ref{LDA}). The results for three values of $\gamma(0)$ at
$t=5\times 10^2$ and various polarizations are shown in Fig.~\ref{T500}. In
Fig.~\ref{T500} the distance $x$ is plotted in units of the Thomas-Fermi
radius in the Gross-Pitaevskii regime \cite{KGDS2}, $R_{GP}=(6 N
c/\omega^2)^{1/3}$ with $N$ the number of particles in the trap
$N=\int (n_1(x)+n_2(x))\, dx$,  and $\omega$
is the trap oscillation frequency. The left panels of Fig.~\ref{T500} depicts
the situation in which the gas in the center of the trap is in the GPa regime
with three polarizations very close to $1$ (see Fig.~\ref{PD} where
$\gamma(0)=1.7\times 10^{-3}$ is the leftmost point on the $t=5\times 10^2$
line) for which $g_2^{(T)}(0)\simeq 1$.  While in the center of the trap the
local correlator is almost constant showing no dependence on $P(0)$ the
situation changes dramatically in the tails where the gas is in the DC
regime.
Here we see that even small variations of the center polarization have
a significant effect in the tail behavior of the local correlator
and polarization.
For larger values of $\gamma(0)$ (center and right
panels of Fig.~\ref{T500}) the effect of the polarization in the center of the
trap on the local correlation function becomes more pronounced with
$g_2^{(T)}(x)$ increasing both as a function of $\gamma$ and $P(0)$. In
addition we see that the density profiles present larger tails as we approach
the DC regime.
\begin{figure}
\includegraphics[width=0.4\linewidth]{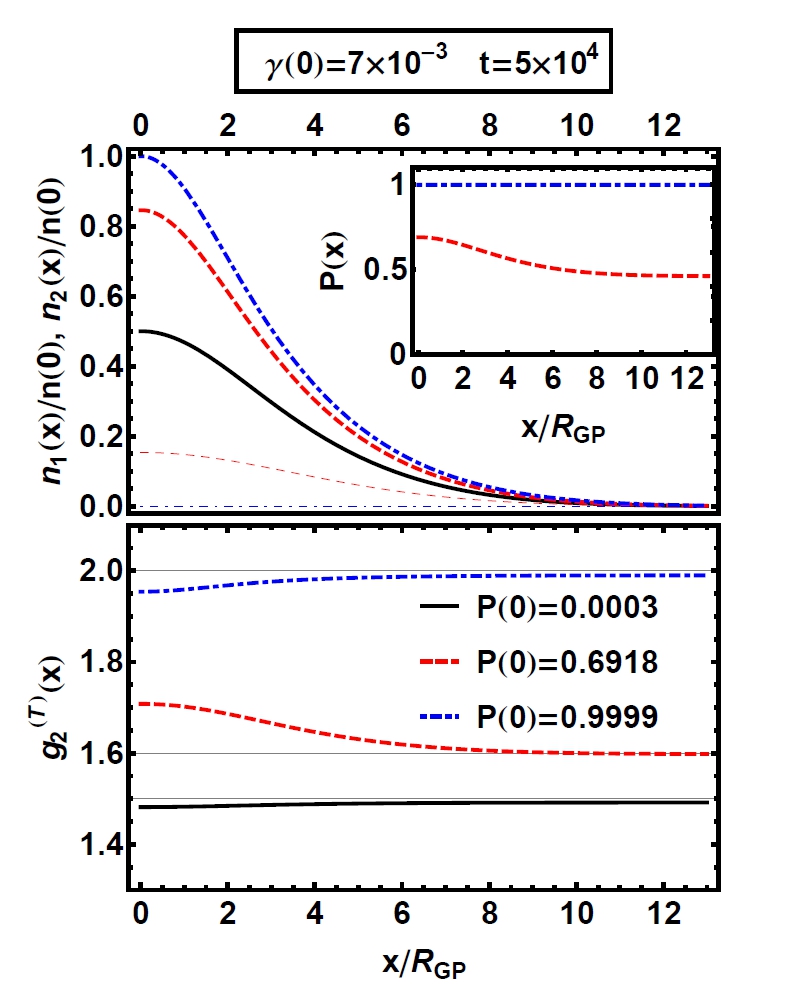}
\caption{(Color online) Density profiles (upper panel) and local density correlation function
  $g_2^{(T)}(x)$ (lower panel) of a trapped 2CBG at temperature $t=5\times
  10^4$ for $\gamma(0)=7\times 10^{-3}$ and polarizations $P(0)=\{0.0003,
  0.6918, 0.9999\}$. In the
  upper panels the thick continuous, thick dashed, and thick dot-dashed (thin continuous,
  thin dashed and thin dot-dashed) lines represent the majority
  (minority) components and $n(0)$ is the total density in the center of the
  trap. The inset depicts the variation of the
  polarization in the trap. In the lower panel the thin lines denote the
  asymptotic values given by Eq.~(\ref{asymptW})}
\label{T50000}
\end{figure}
For large values of $\gamma$ and $t$, when the 2CBG is in the DC regime, it
can be shown using Wick's theorem and Boltzmann distribution \cite{KC,CKB} that
for a uniform system the asymptotic values of the local correlator and
polarization are
\be\label{asymptW}
g_2^{(T)}=1+\frac{e^{2\beta(\mu+H)}+e^{2\beta(\mu-H)}}{(e^{\beta(\mu+H)}+e^{\beta(\mu-H)})^2}\, ,\ \ \
P=\frac{e^{\beta(\mu+H)}-e^{\beta(\mu-H)}}{e^{\beta(\mu+H)}+e^{\beta(\mu-H)}}\, .
\ee
This means that in the case of a trapped system at very large temperatures the
local correlator at the edge of the sample should tend asymptotically to $2$
for $P(0)=1$ (very large magnetic field) and to $1.5$ for $P(0)=0$ (vanishing
magnetic field). This can be seen clearly in Fig.~\ref{T50000} where we
present the density profiles and local correlator for $t=5\times 10^4$ and
$\gamma(0)=7\times 10^{-3}$. Here the entire sample is in the DC regime and
the effect of $P(0)$ on $g_2^{(T)}$ is very large.

\subsection{Density profiles and local correlation functions at intermediate temperatures}

\begin{figure}
\includegraphics[width=0.9\linewidth]{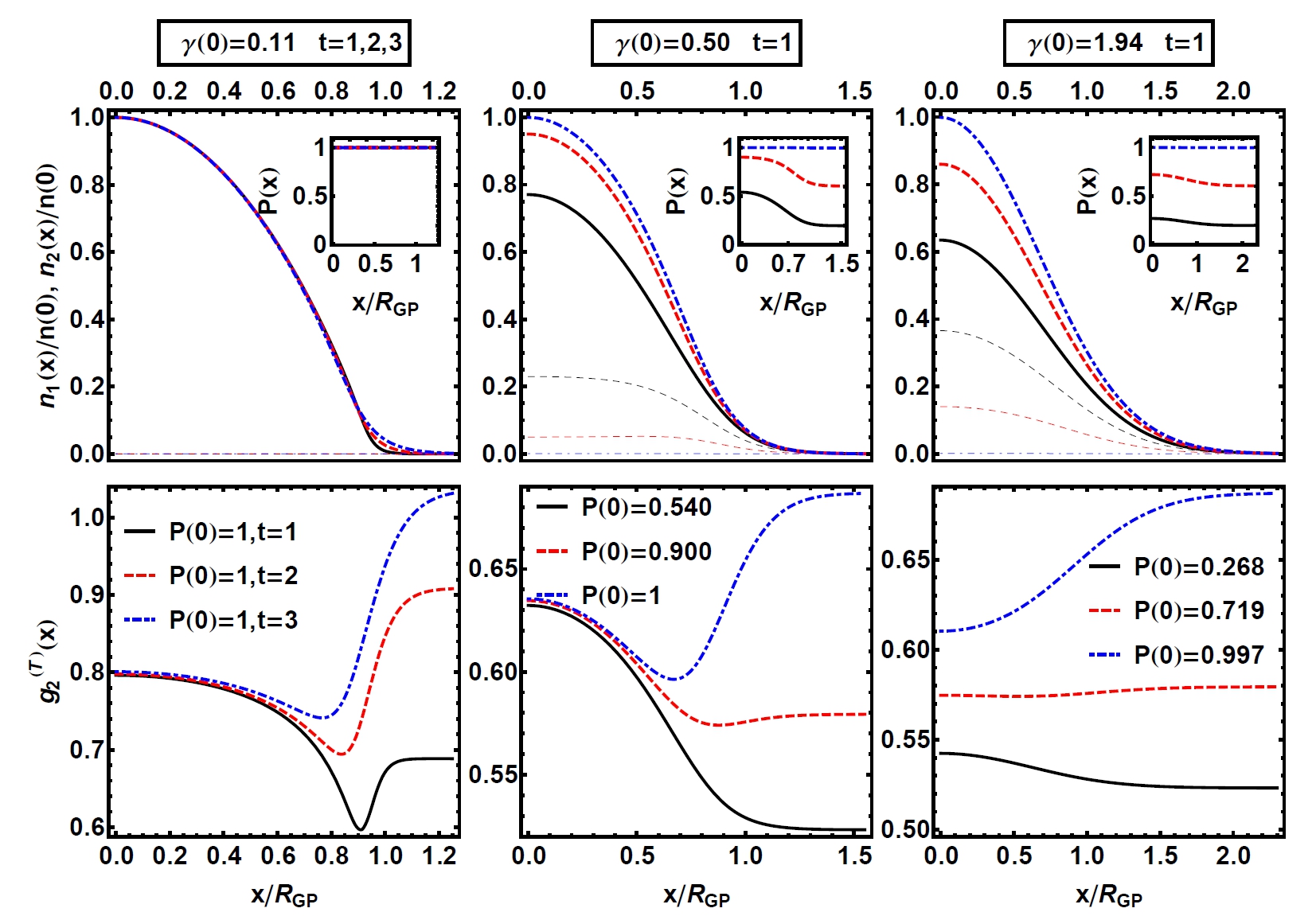}
\caption{(Color online) Density profiles (upper panels) and local density correlation
  function $g_2^{(T)}(x)$ (lower panels) of a trapped 2CBG at different
  temperatures and polarizations for three values of $\gamma(0)=\{0.11, 0.5,
  1.94\}$ and various polarizations. In the
  upper panels the thick continuous, thick dashed, and thick dot-dashed (thin continuous,
  thin dashed and thin dot-dashed) lines represent the majority
  (minority) components and $n(0)$ is the total density in the center of the
  trap. The insets depict the variation of
  the polarization in the trap.  }
\label{Transition}
\end{figure}

For temperatures close to $1$ we can encounter the situation in which for
small variations of $t$ the profile of the local correlation function changes
drastically as it can be seen in the left panels of Fig.~\ref{Transition} for
the fully polarized gas with $\gamma(0)=0.11$ and $t=\{1,2,3\}$. While the
density profiles are almost identical, $g_2^{(T)}(x)$ presents a highly
nonmonotonic behavior with asymptotic values at the edge of the samples
varying from $0.65$ for $t=1$ to $1.05$ for $t=3$. A similar phenomenon can be
seen for fixed temperature and various polarizations in the center and right
panels of Fig.~\ref{Transition}. For $t=1$ and $\gamma(0)=0.5$ the local
correlation functions are almost equal in the center of the system but for
large distances $g_2^{(T)}(x)$ is suppressed for small polarizations and
enhanced for large polarizations. The role of $P(0)$ increases at strong
coupling (large $\gamma(0)$) as shown in the right panel where for
$\gamma(0)=1.94$ and $t=1$ the local correlators present a strong dependence
on the polarization even in the center of the trap.

\subsection{Density profiles and local correlation functions at low-temperatures}

\begin{figure}
\includegraphics[width=0.9\linewidth]{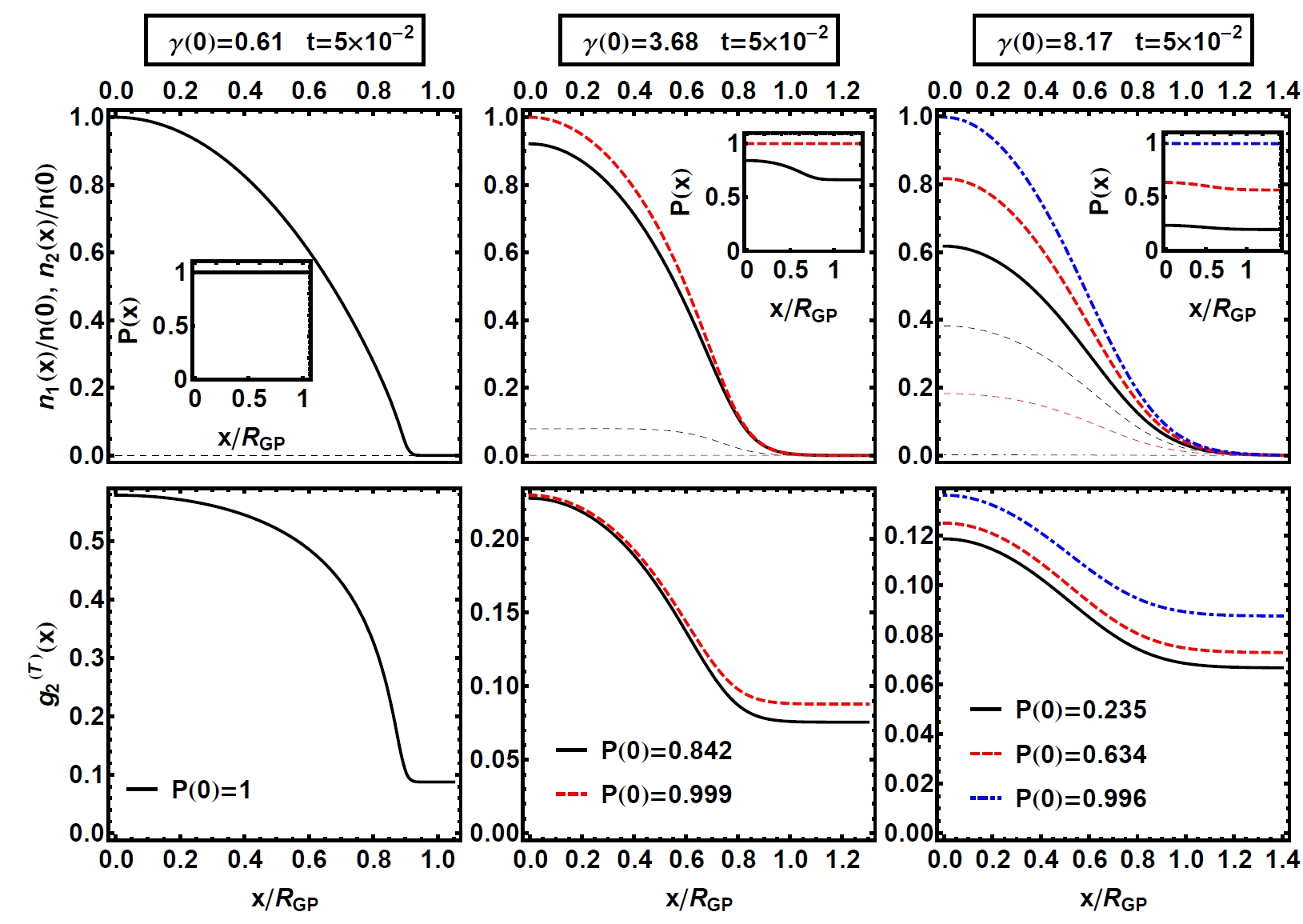}
\caption{(Color online) Density profiles (upper panels) and local density correlation
  function $g_2^{(T)}(x)$ (lower panels) of a trapped 2CBG at temperature
  $t=5\times 10^{-2}$ for three values of $\gamma(0)=\{0.61, 3.68, 8.17\}$ and
  various polarizations. In the
  upper panels the thick continuous, thick dashed, and thick dot-dashed (thin continuous,
  thin dashed and thin dot-dashed) lines represent the majority
  (minority) components and $n(0)$ is the total density in the center of the
  trap. The insets depict the variation of the
  polarization in the trap.  }
\label{T005}
\end{figure}
At low-temperatures ($t\ll 1$) the edges of the system will always be in the
TG regime and, therefore, the local correlator will be strongly suppressed.
In addition, we expect the density profiles to be sharper as it can be seen in
Fig.~\ref{T005} for $t=5\times 10^{-2}$. As expected $g_2^{(T)}(x)$ decreases
in the tails of the distribution until it reaches a limiting value which
depends on $\gamma(0)$.  The influence of the center polarization becomes more
pronounced at stronger coupling with $g_2^{(T)}(0)$ monotonically increasing
as a function of $P(0)$. It should be noted that numerical investigations at
low-temperatures and low-polarizations ($H\rightarrow 0$) are very difficult
due to the fact that in this regime the NLIEs (\ref{NLIEgas}) become
numerically unstable (the same phenomenon happens in the case of the TBA
equations (\ref{TBAeq})) as a consequence of the first-order phase transition
at $T=0\, ,H=0.$ For $t=5\times 10^{-4}$ results are shown in
Fig.~\ref{T00005}.
\begin{figure}
\includegraphics[width=0.9\linewidth]{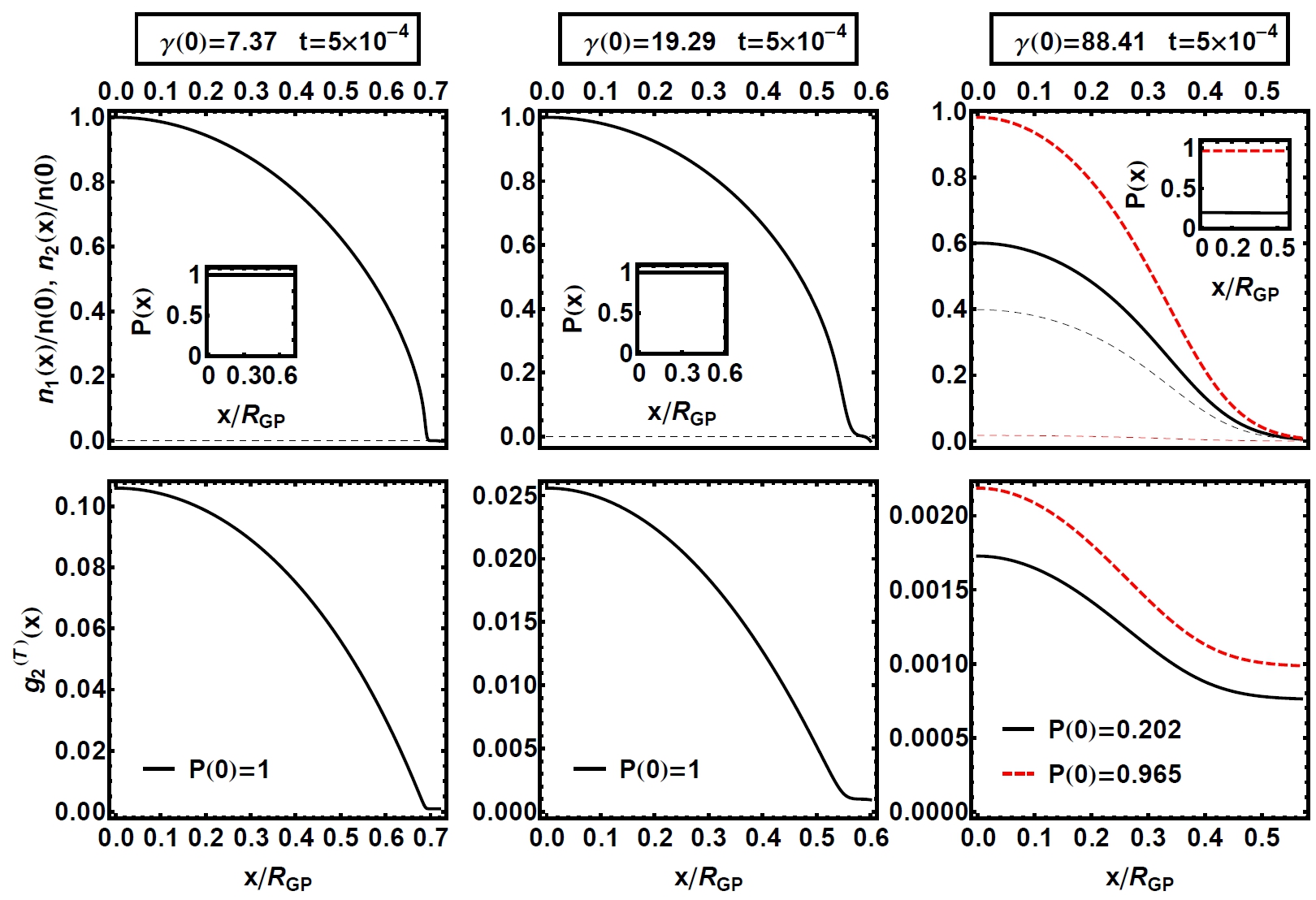}
\caption{(Color online) Density profiles (upper panels) and local density correlation
  function $g_2^{(T)}(x)$ (lower panels) of a trapped 2CBG at temperature
  $t=5\times 10^{-4}$ for three values of $\gamma(0)=\{7.37, 19.29, 88.41\}$
  and various polarizations. In the
  upper panels the thick continuous, thick dashed (thin continuous,
  thin dashed) lines represent the majority
  (minority) components and $n(0)$ is the total density in the center of the
  trap. The insets depict the variation of the
  polarization in the trap.  }
\label{T00005}
\end{figure}

The numerical data presented above show that the polarization at the center of
the trap has a profound influence on the density profiles and the local
density correlation function of the inhomogeneous system.  $g_2^{(T)}(0)$ is a
monotonic increasing function of $P(0)$ with the influence of the center
polarization being more pronounced at strong coupling. Also, the local
correlation function which takes a wide range of values between $0$ and $2$
serves as a better discriminant than the density profiles for the
different regimes of the 2CBG.

\section{The 2CBG as the continuum limit of the  Perk-Schultz spin chain with $(---)$ grading}\label{S2}

Our method of deriving the thermodynamic description of the 2CBG,
Eqs.~(\ref{GCPgas}) and (\ref{NLIEgas}), is based on the fact that the
continuum model can be obtained by performing a specific limit of an
appropriate lattice model for which the quantum transfer matrix technique can
be employed. For the two-component Bose gas the relevant lattice model is the
$q=3$ Perk-Schultz spin chain. In this section we are going to present this
scaling limit and show how the grandcanonical partition function of the 2CBG
can be derived from the canonical partition function of the spin chain.

The Hamiltonian of the $q=3$ Perk-Schultz spin chain with arbitrary grading is \cite{BVV,dV,dVL,L}
\begin{align}\label{Hs}
\cal{H}_{PS}&=J\varepsilon_1\sum_{j=1}^L\left(
\cos\g\sum_{a=1}^3\varepsilon_a e_{aa}^{(j)}e_{aa}^{(j+1)}+
\sum_{\substack{a,b=1\\ a \ne b}}^3
e_{ab}^{(j)}e_{ba}^{(j+1)}
+i\sin\g\sum_{\substack{a,b=1\\ a\ne b}}^3
\mbox{sign}(a-b)e_{aa}^{(j)}e_{bb}^{(j+1)}\right)-
\sum_{j=1}^L\sum_{a=1}^3 h_a e_{aa}^{(j)}\, ,
\end{align}
where $L$ is the number of lattice sites, $\gamma\in(0,\pi)$ determines the
anisotropy, ($q=e^{i\gamma}$), $J>0$ determines the strength of the
interaction and $h_a$ are chemical potentials. The parameters
$(\varepsilon_1,\varepsilon_2,\varepsilon_3)$ can take the values $\pm 1$ and
we will call them the grading of the system. In (\ref{Hs}),
$e^{(j)}_{ab}=\mathbb{I}_3^{\otimes j-1}\otimes e_{ab}\otimes
\mathbb{I}_3^{\otimes L-j}\, ,$ with $e_{ab}$ the $3$-by-$3$ matrix with
elements $(e_{ab})_{ij} =\delta_{a i}\delta_{b j}$ and $\mathbb{I}_3$ the
$3$-by-$3$ unit matrix. The Perk-Schultz Hamiltonian is the sum of
${\cal{H}}_b$ (in the brackets) which is the fundamental spin-model associated
with the trigonometrical Perk-Schultz $\R$-matrix \cite{dV} and ${\cal{H}}_c$
the chemical potential part which does not break the integrability. In the
next section and Appendix \ref{ap11} we will show that for
$(\varepsilon_1,\varepsilon_2, \varepsilon_3)=(---)$ the energy spectrum is
\be\label{energyPS}
E_{PS}=\sum_{j=1}^M e_0(v^{(1)}_j)+M_1(h_2-h_3)+E_0\, ,\ \ \
e_0(v)=J\frac{\sin^2\g}{\sin(v-\g)\sin v}+h_1-h_2\, ,\ \ \ E_0=JL\cos\g-h_1L\, ,
\ee
with $\{v^{(1)}_j\}_{j=1}^M$ satisfying the BAEs
\begin{subequations}\label{BEs}
\begin{align}
& \left(\frac{\sin(\g-v_s^{(1)})}{\sin v_s^{(1)}}\right)^L=\prod_{\substack{j=1\\ j\ne s}}^M\frac{\sin(v_s^{(1)}-v_j^{(1)}-\g)}{\sin(v_s^{(1)}-v_j^{(1)}+\g)}
\prod_{p=1}^{M_1}(-)\frac{\sin(v_s^{(1)}-v_p^{(2)}+\g)}{\sin(v_s^{(1)}-v_p^{(2)})}\, ,\ \ \   s=1,\cdots, M\, ,\\
&\prod_{j=1}^M(-)\frac{\sin(v_l^{(2)}-v_j^{(1)}-\g)}{\sin(v_l^{(2)}-v_j^{(1)})}=\prod_{\substack{p=1\\p\ne l}}^{M_1}\frac{\sin(v_l^{(2)}-v_p^{(2)}-\g)}{\sin(v_l^{(2)}-v_p^{(2)}+\g)}\, ,\ \ \  l=1,\cdots, M_1\, .
\end{align}
\end{subequations}
First, we will show how the BAEs (\ref{BEc}) of the 2CBG can be obtained from
the BAEs (\ref{BEs}) of the Perk-Schultz spin chain. The spin chain is
characterized by the following set of parameters: number of lattice sites $L$,
anisotropy $\g$, strength of the interaction $J>0$, lattice constant $\delta$
and chemical potentials $h_1,h_2,h_3$.
Consider $v_s^{(1)}\rightarrow i\delta
k_s^{(1)}/ \epsilon+\g/2$ and $v_s^{(2)}\rightarrow i\delta
k_s^{(2)}/\epsilon+\g+\pi/2$ with $\epsilon\rightarrow 0$.  The BAEs
(\ref{BEs}) become
\begin{align*}
&\left((-)\frac{\sinh(\delta k_s^{(1)}/\epsilon-i\g/2)}{\sinh(\delta k_s^{(1)}/\epsilon+i\g/2)}\right)^L=
\prod_{\substack{j=1\\j\ne s}}^{M}\frac{\sinh(\delta k_s^{(1)}/\epsilon-\delta k_j^{(1)}/\epsilon -i\g)}{\sinh(\delta k_s^{(1)}/\epsilon-\delta k_j^{(1)}/\epsilon +i\g)}
\prod_{p=1}^{M_1}(-)\frac{\cosh(\delta k_s^{(1)}/\epsilon-\delta k_p^{(2)}/\epsilon+i\g/2)}{\cosh(\delta k_s^{(1)}/\epsilon-\delta k_p^{(2)}/\epsilon-i\g/2)}\, ,\\
&\prod_{j=1}^{M}(-)\frac{\cosh(\delta k_l^{(2)}/\epsilon-\delta k_j^{(1)}/\epsilon+i\g/2)}{\cosh(\delta k_l^{(2)}/\epsilon-\delta k_j^{(1)}/\epsilon-i\g/2)}
=\prod_{\substack{p=1\\p\ne l}}^{M_1}\frac{\sinh(\delta k_l^{(2)}/\epsilon-\delta k_p^{(2)}/\epsilon +i\g)}
{\sinh(\delta k_l^{(2)}/\epsilon-\delta k_p^{(2)}/\epsilon -i\g)}\, .
\end{align*}
In the second step we perform $\g=\pi-\epsilon$ with the result
\begin{subequations}\label{int1}
\begin{align}
&\left(\frac{\cosh(\delta k_s^{(1)}/\epsilon+i\epsilon/2)}{\cosh(\delta k_s^{(1)}/\epsilon-i\epsilon/2)}\right)^L=
\prod_{\substack{j=1\\j\ne s}}^{M}\frac{\sinh(\delta k_s^{(1)}/\epsilon-\delta k_j^{(1)}/\epsilon +i\epsilon)}{\sinh(\delta k_s^{(1)}/\epsilon-\delta k_j^{(1)}/\epsilon -i\epsilon)}
\prod_{p=1}^{M_1}\frac{\sinh(\delta k_s^{(1)}/\epsilon-\delta k_p^{(2)}/\epsilon-i\epsilon/2)}{\sinh(\delta k_s^{(1)}/\epsilon-\delta k_p^{(2)}/\epsilon+i\epsilon/2)}\, ,\\
&\prod_{j=1}^{M}\frac{\sinh(\delta k_l^{(2)}/\epsilon-\delta k_j^{(1)}/\epsilon-i\epsilon/2)}{\sinh(\delta k_l^{(2)}/\epsilon-\delta k_j^{(1)}/\epsilon+i\epsilon/2)}
=\prod_{\substack{p=1\\p\ne l}}^{M_1}\frac{\sinh(\delta k_l^{(2)}/\epsilon-\delta k_p^{(2)}/\epsilon -i\epsilon)}
{\sinh(\delta k_l^{(2)}/\epsilon-\delta k_p^{(2)}/\epsilon +i\epsilon)}\, .
\end{align}
\end{subequations}
Finally, taking the limit $L\rightarrow \infty$ like
$\mathcal{O}(1/\epsilon^2),$ $\delta\rightarrow 0$ like $
\mathcal{O}(\epsilon^2),$ such that $L_B=L\delta$ and $c=\epsilon^2/\delta$,
equating $M_1=M_1^B$ and using
\[
\frac{\cosh(\delta k_s^{(1)}/\epsilon+i\epsilon/2)}{\cosh(\delta k_s^{(1)}/\epsilon-i\epsilon/2)}\sim
\frac{1+ik_s^{(1)}\delta/2}{1-ik_s^{(1)}\delta/2}
\]
we find that (\ref{int1}) transform into the BAEs (\ref{BEc}) of the Bose
gas. Performing the same transformations in the expression for the energy we
obtain
\be\label{int2}
E_{PS}-E_0=\sum_{j=1}^M\left[J\delta^2(k_j^{(1)})^2-J\epsilon^2-J\epsilon^4/4\right]+(h_1-h_2)M+(h_2-h_3)M_1+\mathcal{O}(\epsilon^6)\, ,
\ee
which shows that in order to obtain the energy spectrum of the Bose gas
(\ref{eb}) we need to scale $J\rightarrow \infty$ like
$\mathcal{O}(1/\epsilon^4)$ and $h_1\rightarrow\infty$ like
$\mathcal{O}(1/\epsilon^2)$ such that $J\delta^2=1$ and
$J\epsilon^2+J\epsilon^4/4-h_1+h_2$ is finite. Therefore, in the thermodynamic
limit performing this scaling we have
\be\label{eqz}
e^{\beta E_0}Z(h_1,h_2,h_3,\beta)\rightarrow \mathcal{Z}(\mu,H,\beta)\, ,
\ee
where $\beta=1/T$, $Z(h_1,h_2,h_3,\beta)$ is the canonical partition function
of the spin chain and $\mathcal{Z}(\mu,H,\beta)$ is the grandcanonical
partition function of the 2CBG.

This scaling limit is independent of the temperature of the systems. However,
taking into account that we are interested in the thermodynamic behavior it is
more useful to introduce a scaling limit which involves also the
temperature. This is justified because from the thermodynamics point of view
the energy spectrum enters via the expression $e^{-\beta E}$.  Multiplying
with $\beta$ the identity (\ref{int2}) we find that $\beta (E_{PS}-E_0)=\bar
\beta E_{Bose}$, with $\bar \beta$ the inverse temperature in the continuum
system, if $\beta=\bar\beta/\delta^2$, $J=1$ $h_1\rightarrow 0$ like
$\mathcal{O}(\epsilon^2)$, such that $J\epsilon^2/\delta^2-h_1/\delta^2$ is
finite and $h_2,h_3\rightarrow 0$ like $\mathcal{O}(\epsilon^4)$. In this case
Eq.~(\ref{eqz}) becomes
\be\label{rel1}
e^{\beta E_0}Z(h_1,h_2,h_3,\beta)\rightarrow \mathcal{Z}(\mu,H,\bar\beta)\, .
\ee
In Table \ref{table1} we present in a compact form this continuum limit (the
spectral parameter in the lattice model scales like $v\rightarrow i\delta
k/\epsilon$) which will be used in the next sections to derive the
thermodynamic behavior of the 2CBG from similar result for the $(---)$
Perk-Schultz spin chain.
\begin{table}
\caption{\label{table1} Parameters for the $(---)$ Perk-Schultz spin chain and the spinor Bose gas.}
\begin{center}
\begin{tabular}{|l|l|}
  \hline
 $(---)$ Perk-Schultz spin chain & Spinor  Bose gas \\
  \hline \hline
 lattice constant $\delta\rightarrow\mathcal{O}(\epsilon^2)$            & particle mass $m=1/2$\\
 number of lattice sites $L\rightarrow\mathcal{O}(1/\epsilon^2)$        & physical length $L_B=L\delta$\\
 interaction strength $J>0$                                             & repulsion strength $c=\epsilon^2/\delta$\\
 chemical potential $h_1\rightarrow \mathcal{O}(\epsilon^2)$            & chemical potential $\mu=\frac{J\epsilon^2}{\delta^2}-\frac{h_1}{\delta^2}+\frac{J\epsilon^4}{4\delta^2}+\frac{1}{2\delta^2}(h_3+h_2)$\\
 chemical potentials $h_2,h_3\rightarrow \mathcal{O}(\epsilon^4)$         & magnetic field $H=(h_3-h_2)/(2\delta^2)$\\
 inverse temperature $\beta$                                            & inverse temperature $\overline\beta=\beta \delta^2$\\
 anisotropy $\g=\pi-\epsilon$                            &                                                       \\
  \hline
\end{tabular}
\end{center}
\end{table}

\section{Algebraic Bethe Ansatz and the Quantum Transfer Matrix of the $(---)$ Perk-Schultz spin chain}\label{S3}

In the next section we are going to derive a finite set of NLIEs
characterizing the thermodynamic of the $(---)$ Perk-Schultz spin chain using
the associated quantum transfer matrix. Therefore, it will be useful to recall
some basic notions of algebraic Bethe ansatz \cite{KBI}.

\subsection{ABA for the Hamiltonian}

The Hamiltonian (\ref{Hs}) represents the fundamental spin model (Chap. VI of
\cite{KBI}) associated with the trigonometric $q=3$ Perk-Schultz $\R$-matrix
\cite{dV,dVL} defined by

\be\label{rm}
\R(v,w)=\sum_{a=1}^3\R_{aa}^{aa}(v,w)e_{aa}\, \otimes e_{aa}+
\sum_{\substack{a,b=1\\a\ne b}}^3\R_{ab}^{ab}(v,w)\, e_{aa}\otimes e_{bb}
+\sum_{\substack{a,b=1\\a\ne b}}^3 \R_{ba}^{ab}(v,w)\, e_{ab}\otimes e_{ba}\, ,
\ee
with
\footnote{ In \cite{dV,dVL} the authors considered the more general case
  $\R_{ab}^{ab}(v,w)=G_{ab}\frac{\sin(v-w)} {\sin \g}$ with
  $G_{ab}G_{ba}^{-1}=1$ (no summation). The situation considered in our paper
  corresponds to $G_{ab}=1$}
\be
\R_{aa}^{aa}(v,w)=\frac{\sin[\g+\varepsilon_a(v-w)]}{\sin\g}\, , \ \ \ \R_{ab}^{ab}(v,w)\underset{a\ne b}{=}\frac{\sin(v-w)}{\sin \g}\, ,
\ \ \ \R_{ba}^{ab}(v,w)\underset{a\ne b}{=}e^{isgn(a-b)(v-w)}\, ,
\ee
and $(e_{ab})_{ij}=\delta_{ia}\delta_{jb}$ the canonical basis in the space of
$3$-by-$3$ matrices.  This $9$-by-$9$ $\R$-matrix has the property that
$\R(0,0)=\Pe$, with $\Pe_{b_1\, b_2}^{a_1a_2}= \delta_{a_1b_2}\delta_{a_2b_1}$
the permutation matrix, and satisfies the Yang-Baxter equation
\be\label{ybe}
\sum_{a',b',c'=1}^3 \R_{a'b'}^{a\ b}(v,w)\R_{a'' c'}^{a'\ c}(v,\nu)\R_{b''c''}^{b'\ c'}(w,\nu)
=\sum_{a',b',c'=1}^3 \R_{b'c'}^{b\ c}(w,\nu)\R_{a'c''}^{a\  c'}(v,\nu)\R_{a'' b''}^{a'\  b'}(v,w)\, .
\ee
Let us sketch how we can obtain the Hamiltonian (\ref{Hs}) (more precisely,
the component within the brackets) in the framework of ABA. First, we need to
introduce the $\Lo$-operators defined by
\be
\Lo_j(v,0)=\sum_{a,b,a_1,b_1=1}^3\R_{b\, b_1}^{aa_1}(v, 0)e_{ab}^{(0)}e_{a_1b_1}^{(j)}\, ,
\ \ \ \ \  \Lo_j(v,0)\in \mbox{End}\left((\mathbb{C}^3)^{\otimes(L+1)}\right)\, ,
\ee
where $e_{ab}^{(j)}$ is the canonical basis of operators acting on
$(\mathbb{C}^3)^{\otimes(L+1)}$, {\it i.e.},
$e_{ab}^{(0)}=e_{ab}\otimes\mathbb{I}_3^{\otimes L}\, $ and
$e_{ab}^{(j)}=\mathbb{I}_3\otimes \mathbb{I}_3^{\otimes (j-1)}\otimes
e_{ab}\otimes\mathbb{I}_3^{\otimes(L-j)}$ for $j=1,\cdots,L$. These operators
act on the tensorial product of $\mathbb{C}^3$, which is called the auxiliary
space, and the Hilbert space of the spin chain
$\mathfrak{H}=(\mathbb{C}^3)^{\otimes L}$. In this auxiliary space, using the
definition of the $\R$-matrix (\ref{rm}), the $\Lo$-operators can be
represented as
\be\label{defl}
\Lo_j(v,0)=\left(\begin{array}{ccc} \alpha_1(v,0)e_{11}^{(j)}+\beta(v,0)\left[e_{22}^{(j)}+e_{33}^{(j)}\right] & \g_-(v,0)e_{21}^{(j)} & \g_-(v,0)e_{31}^{(j)}\\
                                      \g_+(v,0)e_{12}^{(j)} & \alpha_2(v,0)e_{22}^{(j)}+\beta(v,0)\left[e_{11}^{(j)}+e_{33}^{(j)}\right] & \g_-(v,0) e_{32}^{(j)}\\
                                      \g_+(v,0)e_{13}^{(j)} & \g_+(v,0)e_{23}^{(j)} & \alpha_3(v,0)e_{33}^{(j)}+\beta(v,0)\left[e_{11}^{(j)}+e_{22}^{(j)}\right]

                    \end{array}
\right)
\ee
where we have introduced
\be
\alpha_i(v,w)=\frac{\sin[\g+\varepsilon_i(v-w)]}{\sin\g}\, ,\ \ \ \beta(v,w)=\frac{\sin(v-w)}{\sin \g}\, ,\ \ \
\g_{\pm}(v,w)=e^{\pm isgn(a-b)(v-w)}\, .
\ee
In (\ref{defl}), $e^{(j)}_{ab}$ now represents the canonical basis of operators acting on
$(\mathbb{C}^3)^{\otimes L}$ and, therefore, the elements of the matrix
represent operators acting on the Hilbert space of the spin chain. The
monodromy matrix is the ordered product of $\Lo$-operators
\be\label{int3}
\T(v)=\Lo_L(v,0)\Lo_{L-1}(v,0)\cdots\Lo_1(v,0)\, ,\ \ \ \ \ \
T(v)=\left(\begin{array}{lcr}
                                T_{11}(v) & T_{12}(v) & T_{13}(v) \\
                                T_{21}(v) & T_{22}(v) & T_{23}(v)\\
                                T_{31}(v) & T_{32}(v) & T_{33}(v)\\
              \end{array}\right)\, ,
\ee
and provides a representation of the Yang-Baxter algebra
\be\label{YBA}
\check\R(v,w)[\T(v)\otimes\T(w)]=[\T(v)\otimes\T(w)]\check\R(v,w)\, ,
\ee
where $\check\R_{b_1\, b_2}^{a_1a_2}(v,w)=(\Pe\R)_{b_1\,
  b_2}^{a_1a_2}(v,w)=\R_{b_1b_2}^{a_2a_1}(v,w)$.  In Eq.~(\ref{YBA}) the
tensor product should be understood as the tensor product of $3$-by-$3$
matrices with operator valued entries as presented in the r.h.s of
(\ref{int3}). Finally, the transfer matrix is defined as the trace of the
monodromy matrix in the auxiliary space
\be\label{transferm}
\tr(v)=\mbox{tr}_{0}\T(v)=T_{11}(v)+T_{22}(v)+T_{33}(v)\, .
\ee
Following \cite{F} we can show that
\be\label{int4}
\mathcal{H}_{PS}=J\varepsilon_1\sin\g\,  \tr^{-1}(0)\tr'(0)-
\sum_{j=1}^L\sum_{a=1}^3 h_a e_{aa}^{(j)}\, .
\ee
Obtaining the eigenvalues of the transfer matrix requires the existence of a
pseudovacuum on which the monodromy matrix acts triangularly. Using
(\ref{defl}), it is easy to see that
\be\label{vacuumt}
|\Omega\rangle=\underbrace{\left(\begin{array}{c} 1\\0\\0 \end{array}\right)\otimes\cdots
\otimes\left(\begin{array}{c} 1\\0\\0 \end{array}\right)}_{\mbox{ L times}}\, ,
\ee
satisfies the requirements of a pseudovacuum, and
\be
\T(v)|\Omega\rangle=\left(\begin{array}{ccc} T_{11}(v)|\Omega\rangle & T_{12}(v)|\Omega\rangle & T_{13}(v)|\Omega\rangle \\
                           T_{21}(v)|\Omega\rangle & T_{22}(v)|\Omega\rangle & T_{23}(v)|\Omega\rangle\\
                           T_{31}(v)|\Omega\rangle & T_{32}(v)|\Omega\rangle & T_{33}(v)|\Omega^\rangle\\

                       \end{array}\right)
                    =\left(\begin{array}{ccc} (\alpha_1(v,0))^L|\Omega\rangle & T_{12}(v)|\Omega\rangle & T_{13}(v)|\Omega\rangle \\
                                              0 & (\beta(v,0))^L|\Omega\rangle & 0\\
                                              0 &  0                  &(\beta(v,0))^L|\Omega\rangle
                    \end{array}\right)\, .
\ee
The energy spectrum of the Perk-Schultz spin chain (\ref{energyPS}) with the
$(---)$ grading is obtained using Eq.~(\ref{int4}) and the eigenvalues of the
transfer matrix (see Appendix \ref{ap11}):
\begin{align}
\tau(v)&=\left(\frac{\sin(\g-v)}{\sin\g}\right)^L\prod_{j=1}^M\frac{\sin(\g-v_j^{(1)}+v)}{\sin(v_j^{(1)}-v)}
+\left(\frac{\sin v}{\sin \g}\right)^L\prod_{j=1}^M\frac{\sin(\g-v+v_j^{(1)})}{\sin(v-v_j^{(1)})}\prod_{i=1}^{M_1}
\frac{\sin(\g-v_i^{(2)}+v)}{\sin(v_i^{(2)}-v)}\nonumber\\
&\ \ \ \ \ \ \ \ \ \ \ \ \ \ \ \ \ \ \ \ \ \ \ \ \ \ \ \ \ \ \ \ \ \ \ \  \ \
+\left(\frac{\sin v}{\sin \g}\right)^L\prod_{j=1}^{M_1}\frac{\sin(\g-v+v_j^{(2)})}{\sin(v-v_j^{(2)})}\, ,
\end{align}
with $\{v_j^{(1)}\}_{j=1}^M, \{v_i^{(2)}\}_{j=1}^{M_1}$ satisfying the BAEs
(\ref{BEs}), and the fact that the contribution of the chemical potential
component is given by $M(h_1-h_2)+M_1(h_2-h_3)+E_0$ \cite{dV}.

\subsection{ABA for the quantum transfer matrix}

The quantum transfer matrix \cite{MS,K1} is an algebraic object which plays a
significant role in the investigation of finite temperature properties of
integrable systems. The importance of the QTM resides in the fact that not
only the free energy can be derived from the largest eigenvalue but also that
different correlation lengths can be characterized as ratios of the leading
eigenvalues. Pedagogical introductions in the subject can be found in
\cite{K3,GKS1,GSuz}.

In order to introduce the QTM we need to define two types of $\Lo$-operators $\Lo_j(\la,-u)\, ,
\tilde\Lo_j(u,\la) \in\mbox{End}\left((\mathbb{C}^3)^{\otimes(N+1)}\right)$ as
\be\label{LQTM}
\Lo_j(\la,-u)=\sum_{a,b,a_1,b_1=1}^3\R_{b\,b_1}^{aa_1}
(\la,-u)e_{ab}^{(0)} e_{a_1b_1}^{(j)}\, ,\ \ \ \ \ \
\tilde\Lo_j(u,\la)=\sum_{a,b,a_1,b_1=1}^3\R_{a_1\, b}^{b_1\, a }
(u,\la)e_{ab}^{(0)} e_{a_1b_1}^{(j)}\, .
\ee
where $u=-J\sin(\varepsilon_1\gamma)\frac{\beta}{N},$ with $N\in 4\mathbb{N}$ the Trotter number, and
$e_{ab}^{(j)}$ is a canonical basis of operators acting on  $(\mathbb{C}^3)^{\otimes(N+1)}.$
In the auxiliary space the $\Lo$-operators  can be represented as
\begin{align}\label{lq}
&\Lo_j(v,-u)=\nonumber\\
&\left(\begin{array}{ccc} \alpha_1(v,-u)e_{11}^{(j)}+\beta(v,-u)\left[e_{22}^{(j)}+e_{33}^{(j)}\right] & \g_-(v,-u)e_{21}^{(j)} &
                                        \g_-(v,-u)e_{31}^{(j)}\\
                                      \g_+(v,-u)e_{12}^{(j)} & \alpha_2(v,-u)e_{22}^{(j)}+\beta(v,-u)\left[e_{11}^{(j)}+e_{33}^{(j)}\right] & \g_-(v,-u) e_{32}^{(j)}\\
                                      \g_+(v,-u)e_{13}^{(j)} & \g_+(v,-u)e_{23}^{(j)} & \alpha_3(v,-u)e_{33}^{(j)}+\beta(v,-u)\left[e_{11}^{(j)}+e_{22}^{(j)}\right]
                    \end{array}
\right)
\end{align}
\begin{align}\label{llq}
&\tilde\Lo_j(u,\ \, v)=\nonumber\\
&\left(\begin{array}{ccc} \ \ \alpha_1(u,\ \, v)e_{11}^{(j)}+\beta(u,\ \, v)\left[e_{22}^{(j)}+e_{33}^{(j)}\right] & \g_+(u,\ \, v)e_{12}^{(j)} &
                                        \g_+(u,\ \, v)e_{13}^{(j)}\\
                                      \g_-(u,\ \, v)e_{21}^{(j)} & \alpha_2(u,\ \, v)e_{22}^{(j)}+\beta(u,\ \, v)\left[e_{11}^{(j)}+e_{33}^{(j)}\right] & \g_+(u,\ v) e_{23}^{(j)}\\
                                      \g_-(u,\ v)e_{31}^{(j)} & \g_-(u,\ v)e_{32}^{(j)} & \alpha_3(u,\ v)e_{33}^{(j)}+\beta(u,\ v)\left[e_{11}^{(j)}+e_{22}^{(j)}\right]\ \
                    \end{array}
\right)
\end{align}
with $e^{(j)}_{ab}$ a canonical basis of operators acting on  $(\mathbb{C}^3)^{\otimes N}$.  The
monodromy matrix of the QTM, which provides another representation of the
Yang-Baxter algebra (\ref{YBA}), is given by
\be\label{mq}
\T^{QTM}(v)=\Lo_{N}(v,-u)\tilde\Lo_{N-1}(u,v)\cdots
\Lo_{2}(v,-u)\tilde\Lo_{1}(u,v)\, .
\ee
Using the representations (\ref{lq}) and (\ref{llq}) we can see that
\be\label{pseudoQTM}
|\Omega\rangle=\underbrace{\left(\begin{array}{c} 1\\
0\\0 \end{array}\right)\otimes\left(\begin{array}{c} 0\\0\\1
\end{array}\right)\otimes\cdots
\otimes\left(\begin{array}{c} 1\\0\\0 \end{array}\right)\otimes
\left(\begin{array}{c} 0\\0\\1 \end{array}\right)}_{\mbox{ N factors}}\, ,
\ee
is a pseudovacuum and
\begin{align*}
\T^{QTM}(v)|\Omega\rangle
&=\left(\begin{array}{ccc} T_{11}^{QTM}(v)|\Omega\rangle & T_{12}^{QTM}(v)|\Omega\rangle & T_{13}^{QTM}(v)|\Omega\rangle \\
                           T_{21}^{QTM}(v)|\Omega\rangle & T_{22}^{QTM}(v)|\Omega\rangle & T_{23}^{QTM}(v)|\Omega\rangle\\
                           T_{31}^{QTM}(v)|\Omega\rangle & T_{32}^{QTM}(v)|\Omega\rangle & T_{33}^{QTM}(v)|\Omega^\rangle\\
  \end{array}\right)\nonumber\\
&=\left(\begin{array}{ccc} (\alpha_1(v,-u)\beta(u,v))^{N/2}|\Omega\rangle & T_{12}^{QTM}(v)|\Omega\rangle & T_{13}^{QTM}(v)|\Omega\rangle \\
                              0 & (\beta(v,-u)\beta(u,v))^{N/2}|\Omega\rangle & T_{23}^{QTM}(v)|\Omega\rangle \\
                              0 &  0                  &(\beta(v,-u)\alpha_3(u,v))^{N/2}|\Omega\rangle
  \end{array}\right)\, .
\end{align*}
The presence of the chemical potential term in the Hamiltonian (\ref{Hs}) is
taken into account via the transformation
\begin{align*}
\T^{QTM}(v)\rightarrow
\T^{QTM}(v)
  \left(\begin{array}{ccc}  e^{\beta h_1} & 0 & 0 \\
                            0 & e^{\beta h_2} & 0\\
                            0 & 0 & e^{\beta h_3}\\
  \end{array}\right)\, .
\end{align*}
The QTM is defined as the trace in the auxiliary space of the monodromy matrix
\be\label{qtransferm}
\tr^{QTM}(v)=\mbox{tr}_0\T^{QTM}(v)=T_{11}^{QTM}(v)+T_{22}^{QTM}(v)+T_{33}^{QTM}(v).
\ee
The existence of the pseudovacuum and the Yang-Baxter algebra (\ref{YBA})
ensures that the eigenvalues of $\tr^{QTM}(v)$ for the grading $(---)$ can be
obtained using ABA with the result (see Appendix \ref{ap11}):
\[
\Lambda^{QTM}(v)=\lambda_1(v)+\lambda_2(v)+\lambda_3(v)\, ,
\]
where
\begin{align}\label{int5}
\lambda_1(v)&=e^{\beta h_3}\left(\frac{\sin(\g-v-u)}{\sin\g}\frac{\sin(u-v)}{\sin\g}\right)^{N/2}\prod_{j=1}^M\frac{\sin(\g-v_j^{(1)}+v)}{\sin(v_j^{(1)}-v)}\, ,\nonumber\\
\lambda_2(v)&
=e^{\beta h_1}\left(\frac{\sin (v+u)}{\sin \g}\frac{\sin(u-v)}{\sin\g}\right)^{N/2}\prod_{j=1}^M\frac{\sin(\g-v+v_j^{(1)})}{\sin(v-v_j^{(1)})}\prod_{i=1}^{M_1}
\frac{\sin(\g-v_i^{(2)}+v)}{\sin(v_i^{(2)}-v)}\, ,\\
\lambda_3(v)&
=e^{\beta h_2}\left(\frac{\sin (v+u)}{\sin \g}\frac{\sin(\g-u+v)}{\sin\g}\right)^{N/2}\prod_{j=1}^{M_1}\frac{\sin(\g-v+v_j^{(2)})}{\sin(v-v_j^{(2)})}\, ,\nonumber
\end{align}
with $\{v_j^{(1)}\}_{j=1}^M$, ($\{v_i^{(2)}\}_{i=1}^{M_1}$) solutions of the BAEs $\lam_1(v^{(1)}_j)/\lam_{2}(v^{(1)}_j)=-1\, ,$
($\lam_2(v^{(2)}_j)/\lam_{3}(v^{(2)}_j)=-1\,$).


\section{Free energy of the $(---)$ Perk-Schultz spin chain }\label{S4}

The free energy per lattice site of the Perk-Schultz spin chain is obtained
from the largest eigenvalue of the QTM, denoted by $\Lambda_0(v)$, via the
relation $f(h_1,h_2,h_3,\beta)=-\ln \Lambda_0(0)/\beta$. The largest eigenvalue lies in
the $(N/2,N/2)$-sector which means that $M,M_1=N/2$ in Eqs.~(\ref{int5}).
For our purposes it is useful to change the spectral parameter $v\rightarrow iv$.
Then, the expression for the largest eigenvalue can be written as ($N\in 4\mathbb{N}$)
\be\label{deflam}
\Lambda_0(v)=\lam_1(v)+\lam_2(v)+\lam_3(v)\, ,\ \ \ \lam_j(v)=\phi_-(v)\phi_+(v)\frac{q_{j-1}(v+i\g)}{q_{j-1}(v)}
\frac{q_{j}(v-i\g)}{q_{j}(v)}e^{\beta \tilde h_j}\, ,
\ee
where $ $ $(\tilde h_1, \tilde h_2, \tilde h_3)= (h_3,h_1,h_2)$ and
\be
\phi_{\pm}(v)=\left(\frac{\sinh(v\pm iu)}{\sin \g}\right)^{N/2}\, , \ \ \ \ \ \ \ \ \
q_j(v)=\left\{\begin{array}{lr}
               \phi_-(v) & j=0\, \\
               \prod_{r=1}^{N/2}\sinh(v-v_{r}^{(j)})& j=1,2\, \\
               \phi_+(v) & j=3\, .\\
              \end{array}\right.
\ee
Using these notations the Bethe equations can be written as
$\lam_j(v^{(j)}_r)/\lam_{j+1}(v^{(j)}_r)=-1\, ,\  r=1,\cdots,N/2\, . $

\subsection{Nonlinear integral equations for the auxiliary functions}

We want to obtain an integral expression for the largest eigenvalue of the QTM
which can be easily implemented numerically. In this paper we are going to
use a method which can be understood as the multicomponent generalization of
the technique presented in \cite{K3} (other thermodynamic descriptions of the
Perk-Schultz spin chain can be found in \cite{PZinn,KWZ,TT1,Tsub1}). We will
introduce a set of two auxiliary functions for which the position of zeroes
and poles is known, allowing the derivation of NLIEs satisfied by these
functions using the Cauchy theorem. In the last step we are going to obtain an
integral expression for the largest eigenvalue in terms of the auxiliary
functions.

Define
\begin{subequations}\label{defaux}
\begin{align}
\au(v)&=\frac{\lam_1(v)}{\lam_2(v)}=\frac{\phi_-(v+i\g)}{\phi_-(v)}\frac{q_1(v-i\g)}{q_1(v+i\g)}\frac{q_2(v)}{q_2(v-i\g)}e^{\beta(h_3-h_1)}\, , \\
\ad(v)&=\frac{\lam_3(v)}{\lam_2(v)}=\frac{\phi_+(v-i\g)}{\phi_+(v)}\frac{q_1(v)}{q_1(v+i\g)}\frac{q_2(v+i\g)}{q_2(v-i\g)}e^{\beta(h_2-h_1)}\, ,
\end{align}
\end{subequations}
with $\mathfrak{a}_j(v)$ periodic of period $i\pi$. Each of the equations
$\mathfrak{a}_j(v)=-1$ has $3N/2$ solutions, of which $N/2$ are the Bethe
roots $\{v^{(j)}_r\}_{r=1}^{N/2}$ and $N$ solutions which are called holes and
will be denoted by $ \{v'^{(j)}_r\}_{r=1}^N$. The distribution of Bethe roots
and holes characterizing the largest eigenvalue of the QTM for
$\g\in(0,\pi/2)$ is presented in Fig.~\ref{LEfig}.
\begin{figure}
\includegraphics[width=0.5\linewidth]{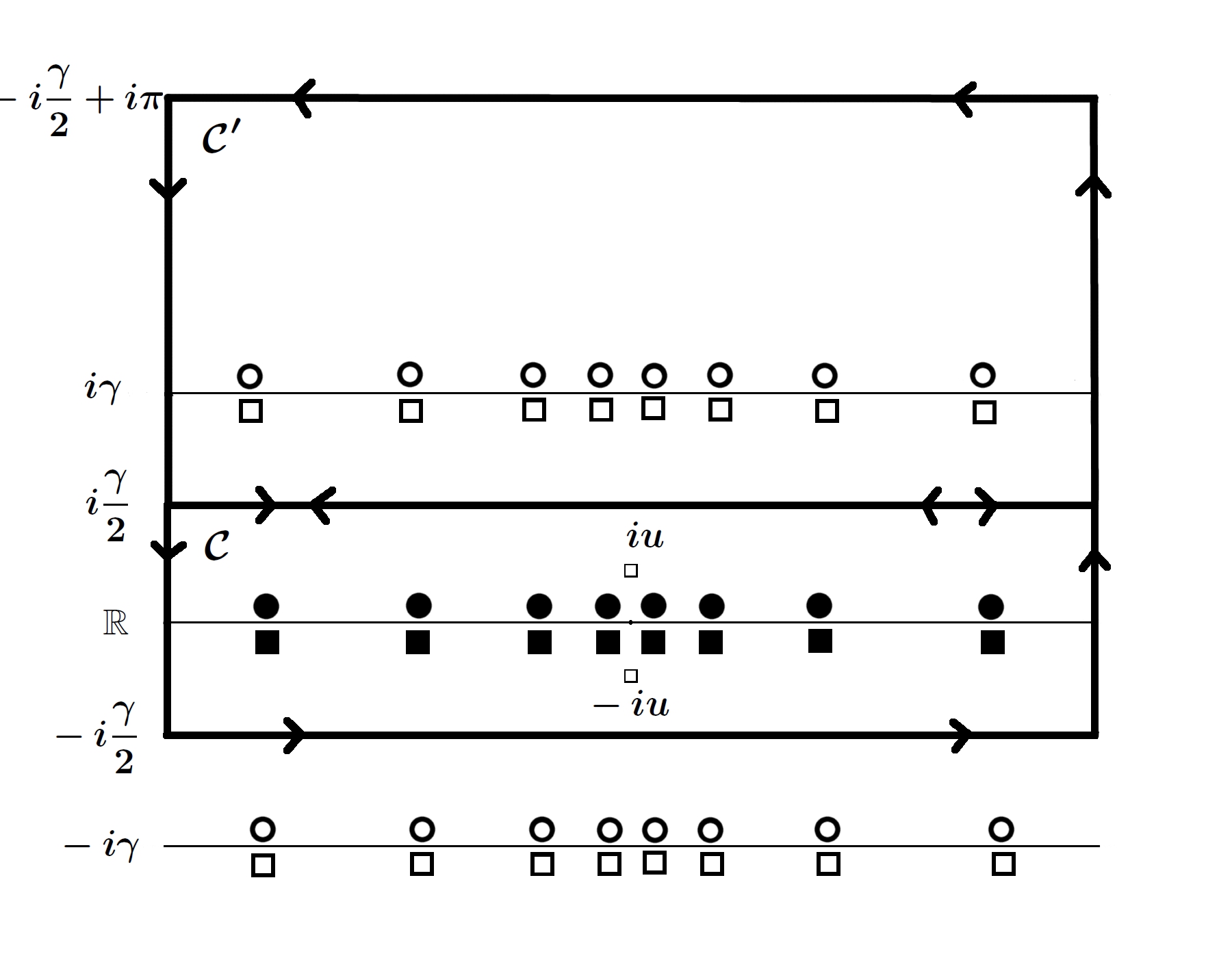}
\caption{Distribution of Bethe roots $\{v^{(1)}_i\}\, ,\{v^{(2)}_j\}$
  ($\bullet,\blacksquare$) and holes $\{v'^{(1)}_i\}\, ,\{v'^{(2)}_j\}$
  ($\circ,\square$) characterizing the largest eigenvalue of the QTM
  ($\g\in(0,\pi/2)$). The contour $\mathcal{C}$ contains all the Bethe roots
  and the poles of the auxiliary functions $\mathfrak{a}_j(v)\, $ at $\pm
  iu$. The lower edge of the contour $\mathcal{C}'$ (see Appendix \ref{ap22}
  and next section) coincides with the upper edge of $\mathcal{C}$ but has
  opposite orientation.}
\label{LEfig}
\end{figure}
Let $\mathcal{C}$ be a rectangular contour with the upper (lower) edges
parallel with the real axis through $\pm i(\g-\varepsilon)/2,$ with
$\varepsilon\rightarrow 0$ (please note that the contour is independent of
$N$, the Trotter number).  Inside the contour $\mathcal{C}$ the function
$1+\au(v)\, (1+\ad(v))$ has $N/2$ zeroes at the Bethe roots
$\{v^{(1)}_r\}_{r=1}^{N/2}\, (\{v^{(2)}_r\}_{r=1}^{N/2})$ and a pole of order
$N/2$ at $ iu\, (-iu)$.  Then, for $v$ outside of the contour, and $j=\{1,2\}$
we can introduce
\be\label{f}
f_j(v)\equiv\frac{1}{2\pi i}\int_{\mathcal{C}}\frac{d}{dv}\left(\ln\sinh(v-w)\right)
\ln(1+\mathfrak{a}_j(w))dw=\frac{1}{2\pi i}\int_{\mathcal{C}}\ln\sinh(v-w)
\frac{\mathfrak{a}_j'(w)}{1+\mathfrak{a}_j(w)}dw\, ,
\ee
where the r.h.s of (\ref{f}) was obtained by taking into consideration that
$\ln(1+\mathfrak{a}_j(v))$ has no winding number due to the fact that the
number of zeroes inside the contour is equal to the order of the poles at
$\pm iu$. We can evaluate $f_j(v)$ making use of the following theorem:
\begin{thm}\label{T}\cite{WW}
Let $g(\la)$ be an analytic function inside and on an arbitrary contour
$\mathcal{C}$ in the complex plane. Let $\phi(\la)$ be another function which
is meromorphic inside and on the contour. Denoting the zeros of $\phi(\la)$ in
the interior of $\mathcal{C}$ by $a_1,a_2,\cdots$ (with multiplicities
$r_1,r_2,\cdots$ ) and the poles by $b_1,b_2,\cdots$ (with multiplicities
$s_1,s_2,\cdots$ ), then
\[
\frac{1}{2\pi i}\int_{\mathcal{C}}g(\la)\frac{\phi'(\la)}{\phi(\la)}d\la=
\sum_{i\in\mbox{zeros}}r_ig(a_i)
-\sum_{i\in\mbox{poles}}s_ig(b_i)\, ,
\]
\end{thm}
with the result
\begin{subequations}\label{int6}
\begin{align}
f_1(v)&=\ln q_1(v)-\ln \phi_-(v)-\frac{N}{2}\ln\sin\g\, ,\\
f_2(v)&=\ln q_2(v)-\ln \phi_+(v)-\frac{N}{2}\ln\sin\g\, .
\end{align}
\end{subequations}
Taking the logarithm of (\ref{defaux}) and using (\ref{int6}) we find
\begin{subequations}\label{int7}
\begin{align}
\ln \au(v)=\beta(h_3-h_1)+\ln\left(\frac{\phi_-(v-i\g)}{\phi_+(v-i\g)}\frac{\phi_+(v)}{\phi_-(v)}\right)
+f_1(v-i\g)-f_1(v+i\g)+f_2(v)-f_2(v-i\g)\, ,\\
\ln \ad(v)=\beta(h_2-h_1)+\ln\left(\frac{\phi_+(v+i\g)}{\phi_-(v+i\g)}\frac{\phi_-(v)}{\phi_+(v)}\right)
+f_1(v)-f_1(v+i\g)+f_2(v+i\g)-f_2(v-i\g)\, ,
\end{align}
\end{subequations}
Eqs.~(\ref{int7}) are nonlinear integral equations of convolution type for
$\au(v)$ and $\ad(v)$ which are valid for all $N$. Using the identity
\[
\lim_{N\rightarrow\infty}\ln\left(\frac{\phi_+(v)}{\phi_-(v)}\right)=i\beta J\sin\g\coth v\, ,
\]
we can perform the Trotter limit, $N\rightarrow \infty$, obtaining
\begin{subequations}\label{nlie}
\begin{align}
\ln\au(v)&=\beta(h_3-h_1)-\beta \frac{J\sinh^2(i\g)}{\sinh v\sinh(v-i\g)}+\int_{\mathcal{C}}K_0(v-w)\ln(1+\au(w))\, dw
-\int_{\mathcal{C}}K_2(v-w)\ln(1+\ad(w))\, dw,\label{nlie1}\\
\ln\ad(v)&=\beta(h_2-h_1)-\beta \frac{J\sinh^2(i\g)}{\sinh v\sinh(v+i\g)}+\int_{\mathcal{C}}K_1(v-w)\ln(1+\au(w))\, dw
-\int_{\mathcal{C}}K_0(v-w)\ln(1+\ad(w))\, dw,\label{nlie2}
\end{align}
\end{subequations}
where
\be
K_0(v)=\frac{1}{2\pi i}\frac{\sinh(2i\g)}{\sinh(v+i\g)\sinh(v-i\g)}\, ,\ \
K_1(v)=\frac{1}{2\pi i}\frac{\sinh(i\g)}{\sinh(v+i\g)\sinh v}\, ,\ \
K_2(v)=\frac{1}{2\pi i}\frac{\sinh(i\g)}{\sinh(v-i\g)\sinh v}\, .
\ee
Even though Eqs.~(\ref{nlie}) were derived assuming $\g\in(0,\pi/2)$ they
remain valid also for $\g\in(\pi/2,\pi)$ by replacing $\mathcal{C}$ with a
similar rectangular contour with the upper (lower) edges situated at $\pm
i(\pi-\g-\varepsilon)/2$ with $\epsilon\rightarrow 0$.  An observation is in
order. The above results were derived assuming that $v$ and $v\pm i\g$ were
situated outside the contour $\mathcal{C}$. For $v$ on the real axis the
integrands that appear in the definition of the functions $f_j(v)$ present an
additional pole at $v$ which means that in this case we have to add
$\ln(1+\ad(v))$ on the r.h.s of Eq.~(\ref{nlie1}) and $\ln(1+\au(v))$ on the
r.h.s of Eq.~(\ref{nlie2}).

\subsection{Integral expression for the largest eigenvalue}\label{IntegralE}

The free energy of the spin chain is obtained from the largest eigenvalue of
the QTM evaluated at zero. Taking into account that $\Lambda_0(v)$ is
analytical in a strip around the real axis then is sufficient to obtain an
integral representation for $\Lambda_0(v_0)$ with $v_0$ inside the strip and
then take the limit $v_0\rightarrow 0.$ Choosing $v_0=-iu$, for which
$\phi_+(v_0)=0$, and using (\ref{i1}) we find
\be\label{i2}
\Lambda_0(v_0)=c\ \frac{\phi_-(v_0)q_2^{(h)}(v_0)}{q_1(v_0)}\, .
\ee
Taking the logarithm of (\ref{i2}) and using (\ref{zz1}) we obtain (all
constants are denoted by $c$)
\begin{align}\label{i4}
\ln\Lambda_0(v_0)&=-\ln q_1(v_0)+\ln\phi_-(v_0))+\ln q_2^{(h)}(v_0)+c\, ,\nonumber\\
                  &=\ln q_1^{(h)}(v_0)+\ln q_2^{(h)}(v_0)-\ln q_1(v_0+i\g)-\ln q_2(v_0-i\g)-\ln(1+\mathfrak{a}_1(v_0))+c\, .
\end{align}
Now we need integral representations for $\ln q_1^{(h)}(v_0)$ and $\ln
q_2^{(h)}(v_0)$ with $v_0$ close to the real axis.  For $v$ inside the contour
$\mathcal{C}$ consider
\be
\frac{1}{2\pi i}\int_{\mathcal{C}}d(v-w)\frac{\mathfrak{a}_1'(w)}{1+\mathfrak{a}_1(w)}dw
=-\frac{1}{2\pi i}\int_{\mathcal{C'}}d(v-w)\frac{\mathfrak{a}_1'(w)}{1+\mathfrak{a}_1(w)}dw\, ,
\ee
where on the r.h.s we have used (\ref{integralr}). The latter expression can
be evaluated with the use of Theorem \ref{T} and using that
$1+\mathfrak{a}_1(v)$ is periodic of period $i\pi$ and inside the contour
$\mathcal{C'}$ (see Appendix \ref{ap22}) has $N$ zeros which are holes
$\{v'^{(1)}_j\}_{j=1}^N$ and $N$ poles situated at
$\{v_j^{(1)}-i\g\}_{j=1}^{N/2}$ and $\{v_j^{(2)}+i\g\}_{j=1}^{N/2}$ (some
modulo $i\pi$). We find
\be\label{i3}
\frac{1}{2\pi i}\int_{\mathcal{C}}d(v-w)\frac{\mathfrak{a}_1'(w)}{1+\mathfrak{a}_1(w)}dw=
-\sum_{j=1}^Nd(v-v'^{(1)}_j)+\sum_{j=1}^{N/2}d(v-v_j^{(1)}+i\g)+\sum_{j=1}^{N/2}d(v-v_j^{(2)}-i\g)\, .
\ee
%
%
%
For $v$ inside the contour $\mathcal{C},$  $v+i\g$ is situated outside the contour. Then
\be\label{int10}
\frac{1}{2\pi i}\int_{\mathcal{C}}d(v-w+i\g)\frac{\mathfrak{a}_1'(w)}{1+\mathfrak{a}_1(w)}dw=
\sum_{j=1}^{N/2}d(v-v_j^{(1)}+i\g)-Nd(v-iu+i\g)/2\, ,
\ee
where we have used Theorem \ref{T} and the fact that inside the contour the
function $1+\au(v)$ has $N/2$ zeros at the Bethe roots
$\{v_j^{(1)}\}_{j=1}^{N/2}$ and a pole of order $N/2$ at $iu$.  Taking the
difference of (\ref{i3}) and (\ref{int10}), integrating by parts w.r.t $w$ and
then integrating everything w.r.t. $v$ we obtain
\begin{align}\label{int11}
\frac{1}{2\pi i}\int_{\mathcal{C}}[d(v-w)-d(v-w+i\g)]\ln(1+\au(w))\, dw&=\ln q_2(v-i\g)-\ln q_1^{(h)}(v)+\ln\phi_-(v+i\g)+c\, .
\end{align}
In an analogous fashion we find
\begin{align}\label{int12}
\frac{1}{2\pi i}\int_{\mathcal{C}}[d(v-w)-d(v-w-i\g)]\ln(1+\ad(w))\, dw&=\ln q_1(v+i\g)-\ln q_2^{(h)}(v)+\ln\phi_+(v-i\g)+c\, .
\end{align}
Replacing $\ln q_1^{(h)}(v_0)$ and $\ln q_2^{(h)}(v_0)$ appearing in the second identity of (\ref{i4})  with
the expressions which can be derived from (\ref{int11}) and (\ref{int12}) we find
\begin{align}\label{int13}
\ln\Lambda_0(v_0)=\ln(\phi_+(v_0-i\g)\phi_-(v_0+i\g))-&\int_{\mathcal{C}}K_1(v_0-w)\ln(1+\au(w))\, dw\nonumber\\
&+\int_{\mathcal{C}}K_2(v_0-w)\ln(1+\ad(w))\, dw-\ln(1+\mathfrak{a}_1(v_0))+c\, .
\end{align}
Eq.~(\ref{int13}) is in fact valid for $v_0$ in a narrow strip around the real
axis as a consequence of the analyticity of the largest eigenvalue. Obtaining
the constant of integration requires the introduction of
\be\label{i5}
\bar\Lambda_0(v)=\frac{\Lambda_0(v)}{\phi_+(v-i\g)\phi_-(v+i\g)}\, ,
\ee
which has the property of having constant asymptotics at infinity
$\lim_{v\rightarrow \infty}\bar\Lambda_0(v)=e^{\beta h_1}+e^{\beta
  h_2}+e^{\beta h_3}$.  All the previous considerations in this section are
also valid for $\bar\Lambda_0(v)$ which means that
\begin{align}\label{int14}
\ln\bar\Lambda_0(v)&=c-\int_{\mathcal{C}}K_1(v-w)\ln(1+\au(w))\, dw
+\int_{\mathcal{C}}K_2(v-w)\ln(1+\ad(w))\, dw-\ln(1+\mathfrak{a}_1(v))\, ,
\end{align}
with $c$ a constant and $v$ in a strip around the real axis.  Taking the limit
$v\rightarrow\infty$ and using $\lim_{v\rightarrow \infty}
\au(v)=e^{\beta(h_3-h_1)}$ and $\lim_{v\rightarrow \infty}
\ad(v)=e^{\beta(h_2-h_1)}$ we find that
\be\label{c}
c=\ln\left[\frac{e^{\beta h_1}+e^{\beta h_2}+e^{\beta h_3}}{1+e^{\beta(h_2-h_1})}\right]
=\beta h_1+\ln\left[\frac{1+e^{\beta (h_2-h_1)}+e^{\beta (h_3-h_1)}}{1+e^{\beta(h_2-h_1})}\right]\, ,
\ee
The last expression for the constant shows that in the scaling limit where
$\beta=\bar \beta/\delta^2$ with $\delta\rightarrow O(\epsilon^2)$ and
$h_1\rightarrow O(\epsilon^2),\ h_{2,3}\rightarrow O(\epsilon^4)$ the term
containing the square parenthesis vanishes.  Finally, from (\ref{i5}) and
using $\lim_{\substack{v\rightarrow 0,}{N\rightarrow
    \infty}}\ln(\phi_+(v-i\g)\phi_-(v+i\g))=-J\beta\cos\g$ we obtain the
integral expression for the largest eigenvalue
\be\label{le}
\ln\Lambda_0(0)=c-J\beta\cos\g-\int_{\mathcal{C}}K_2(w)\ln(1+\au(w))\, dw+\int_{\mathcal{C}}K_1(w)\ln(1+\ad(w))\, dw-\ln(1+\mathfrak{a}_1(0))\, ,
\ee
with $c$ defined in (\ref{c}). This expression, which was derived assuming
$\g\in(0,\pi/2)$, remains valid also for $\g\in(\pi/2,\pi)$ by replacing
$\mathcal{C}$ with a similar rectangular contour with the upper (lower) edges
situated at $\pm i(\pi-\g-\varepsilon)/2$.

\section{Continuum limit}\label{S5}

Having derived an integral expression for the free energy of the spin chain
all that remains in order to obtain the thermodynamic description of the 2CBG
is to perform the continuum limit presented in Section \ref{S2}. In the
continuum limit we can see from Table \ref{table1} that $\g=\pi-\varepsilon$
which means that the contour $\mathcal{C}$ appearing in Eqs.~(\ref{nlie}) and
(\ref{le}) has the upper (lower) edge parallel with the real axis situated at
$\pm i(\pi-\g-\varepsilon)/2.$ We will denote the upper edge of the contour
situated at $i(\pi-\g)/2$ (the $\varepsilon$ term is irrelevant for the
following discussion) by $\mathcal{C}_+$ and by $\mathcal{C}_-$ the lower edge
situated at $-i(\pi-\g)/2$.  For $v\in\mathcal{C}_-\, ,$ $v=x-i(\pi-\g)/2$
with $x$ real, the driving term in the r.h.s. of (\ref{nlie1}) is negative and
equal with
\[
\beta(h_3-h_1)- \beta\frac{J\sin^2\g}{\cosh(x+i\g/2)\cosh(x-i\g/2)}\, ,\ \ \ \ h_3<h_1\, ,J>0\, ,
\]
which means that in the scaling limit which implies $\beta\rightarrow\infty\,
,$ the auxiliary function $\au(v)$ is very small on $\mathcal{C}_-$ and it can
be neglected (for a rigorous justification see \cite{DGK}). In a similar
fashion, for $v\in\mathcal{C}_+\, ,$ $v=x+i(\pi-\g)/2$ with $x$ real, the
driving term in the r.h.s. of (\ref{nlie2}) is negative and equal with
\[
\beta(h_2-h_1)- \beta\frac{J\sin^2\g}{\cosh(x+i\g/2)\cosh(x-i\g/2)}\, ,\ \ \ \ h_2<h_1\, ,J>0\, ,
\]
which means that the contribution of the auxiliary function $\ad(v)$  on $\mathcal{C}_+$
is also negligible. Therefore, in the scaling limit Eqs.~(\ref{nlie}) take the form
\begin{subequations}
\begin{align}
\ln\au(v)&=\beta(h_3-h_1)-\beta \frac{J\sinh^2(i\g)}{\sinh v\sinh(v-i\g)}+\int_{\mathcal{C}_+}K_0(v-w)\ln(1+\au(w))\, dw
-\int_{\mathcal{C}_-}K_2(v-w)\ln(1+\ad(w))\, dw,\\
\ln\ad(v)&=\beta(h_2-h_1)-\beta \frac{J\sinh^2(i\g)}{\sinh v\sinh(v+i\g)}+\int_{\mathcal{C}_+}K_1(v-w)\ln(1+\au(w))\, dw
-\int_{\mathcal{C}_-}K_0(v-w)\ln(1+\ad(w))\, dw.
\end{align}
\end{subequations}
We can shift the free argument $v$ and variable of integration to the line
$+i\g/2$ for the function $\au(v)$ and to the line $-i\g/2$ for the function
$\ad(v)$ without crossing any poles of the driving terms obtaining
\begin{subequations}\label{int15}
\begin{align}
\ln\au(v+i\g/2)=\beta(h_3-h_1)&-\beta \frac{J\sinh^2(i\g)}{\sinh (v+i\g/2)\sinh(v-i\g/2)}-\int_{\mathbb{R}}K_0(v-w)\ln(1+\au(w+i\g/2))\, dw\nonumber\\
&\ \ \ \ \ \ \ \ \ \ \ \ \ \  \ \ \ \ \ \ \  \ \ \ \ \ \ \ \ -\int_{\mathbb{R}}K_2(v-w+i\g-i\varepsilon)\ln(1+\ad(w-i\g/2))\, dw,\\
\ln\ad(v-i\g/2)=\beta(h_2-h_1)&-\beta \frac{J\sinh^2(i\g)}{\sinh (v+i\g/2)\sinh(v-i\g/2)}-\int_{\mathbb{R}}K_1(v-w-i\g+i\varepsilon)\ln(1+\au(w+i\g/2))\, dw\nonumber\\
&\ \ \ \ \  \ \ \ \ \ \ \ \ \ \ \ \ \ \ \ \ \ \ \ \ \ \ \ \ \ -\int_{\mathbb{R}}K_0(v-w)\ln(1+\ad(w-i\g/2))\, dw.
\end{align}
\end{subequations}
where we took into account the negative orientation of $C_+$.  Performing
similar transformations in the integral expression for the largest eigenvalue
(\ref{le}) we find
\be
\ln\Lambda_0(0)=c-J\beta\cos\g+\int_{\mathbb{R}}K_2(w+i\g/2)\ln(1+\au(w+i\g/2))\, dw+\int_{\mathbb{R}}K_1(w-i\g/2)\ln(1+\ad(w-i\g/2))\, dw-\ln(1+\mathfrak{a}_1(0))\, .
\ee

In the continuum limit we have $\g=\pi-\epsilon$, $v\rightarrow \delta
k/\epsilon,$ $ w\rightarrow \delta k'/\epsilon$ (we do not need the $i$ factor
because we have already considered the largest eigenvalue at $iv$, see the
remark before (\ref{deflam})) and
\be
K_0(v)\rightarrow -\frac{\epsilon}{\delta}\frac{1}{2\pi} \frac{2c}{k^2+c^2}\, ,\ \
K_1(v-i\gamma)\rightarrow -\frac{\epsilon}{\delta}\frac{1}{2\pi }\frac{c}{k(k+ic)}\, ,\ \
K_2(v+i\gamma)\rightarrow -\frac{\epsilon}{\delta}\frac{1}{2\pi }\frac{c}{k(k-ic)}\, .
\ee
Introducing $a_1(k)=\au(\delta k/\epsilon+i\g/2),$ $a_2(k)=\ad(\delta
k/\epsilon-i\g/2)$ and taking the scaling limit in Eqs.~(\ref{int15}) we
obtain the NLIEs for the two-component Bose gas, Eqs.~(\ref{NLIEgas}).

The grandcanonical potential per unit length for the Bose gas is, see
Eq.~(\ref{rel1}), $\phi(\mu,H,\bar
\beta)=(f(h_1,h_2,h_3,\beta)-E_0/L)/\delta^3$ with
$f(h_1,h_2,h_3,\beta)=-\ln\Lambda_0(0)/\beta$ the free energy of the spin
chain per lattice site and $E_0/L=J\cos\g-h_1$ the zero point energy. In the
continuum limit the real part of the driving term on the r.h.s of
Eq.~(\ref{nlie1}) becomes large and negative like $\mathcal{O}(1/\epsilon^2)$
which means that $\ln(1+\au(0))/\delta$ vanishes in this limit.  Noticing that
$K_1(\delta k/\epsilon -i\g/2)= K_2(\delta k/\epsilon +i\g/2)\sim\epsilon/
2\pi$ with $dw= \delta dk/\epsilon$ and that the second term in the r.h.s of
(\ref{c}) vanishes in the same limit we obtain Eq.~(\ref{GCPgas}) for the
grandcanonical potential.

\section{Conclusions}

In this article we have investigated the density profiles and local density
correlation functions of the inhomogeneous 2CBG at finite temperature. Our
results derived using a new and numerically efficient solution for the
thermodynamics of the uniform system and the local density approximation
predict that the polarization at the center of the trap has a significant
influence on the local correlator throughout the sample this effect becoming
more pronounced at strong-coupling.  Even though our analysis of the density
profiles did not find any evidence of phase separation this scenario cannot be
fully excluded in the region of low-temperatures and small polarizations where our
equations (and also the TBA result) become numerically unstable as a result of
the first-order phase transition. The method employed in the derivation of the
thermodynamics of the uniform system can also be used in the case of
two-component fermions. In this case the relevant lattice model is the $(-++)$
Perk-Schultz spin chain, and will be the subject of a future publication
\cite{PK1}.

\section{Acknowledgments}

Financial support from the VolkswagenStiftung and the PNII-RU-TE-2012-3-0196
grant of the Romanian National Authority for Scientific Research is gratefully
acknowledged.

\appendix

\section{ABA solution of the generalized $q=3$ Perk-Schultz model}\label{ap11}

In this appendix we present the solution of the generalized $q=3$ Perk-Schultz
model.  The eigenvalues of the transfer matrix (\ref{transferm}) and quantum
transfer matrix (\ref{qtransferm}) can be obtained as particular cases of this
general solution. We should mention that the transfer matrix (\ref{transferm})
was diagonalized in \cite{BVV} and the eigenvalues of the quantum transfer
matrix (\ref{qtransferm}) were conjectured in \cite{KWZ}. For a proof of the
results presented below see \cite{Sm,G1,GS,ACDFR,BR}.

The generalized $q=3$ Perk-Schultz model is the set of all linear
representations of the Yang-Baxter algebra
\be
\check\R(v,w)[\T(v)\otimes\T(w)]=[\T(v)\otimes\T(w)]\check\R(v,w)\, ,
\ee
with the $\R$-matrix defined in (\ref{rm}) and $\T(v)$ the monodromy matrix
which acts triangularly on a pseudovacuum (highest vector) $|\Omega\rangle$
\be
\T(v)|\Omega\rangle=\left(\begin{array}{ccc} A(v)|\Omega\rangle & B_1(v)|\Omega\rangle & B_2(v)|\Omega\rangle \\
                                             C_1(v)|\Omega\rangle & D_{11}(v)|\Omega\rangle & D_{12}(v)|\Omega\rangle\\
                                             C_2(v)|\Omega\rangle & D_{21}(v)|\Omega\rangle & D_{22}(v)|\Omega\rangle\\
                       \end{array}\right)
                    =\left(\begin{array}{ccr} \pu(v)|\Omega\rangle & B_1(v)|\Omega\rangle & B_2(v)|\Omega\rangle \\
                                              0 & \pd(v)|\Omega\rangle & D_{12}(v)|\Omega\rangle\\
                                              0 &  0                  &\pt(v)|\Omega\rangle
                    \end{array}\right)\, .
\ee
The functions $\varphi_j(v)\, ,j=1,2,3$ are called the parameters of the
model. The monodromy matrices of the transfer matrix (\ref{int3}) and QTM
(\ref{mq}) satisfy these requirements with pseudovacua (\ref{vacuumt}) and
(\ref{pseudoQTM}) respectively. The parameters of the model are given by
\be
\pu(v)=(\alpha_1(v,0))^L\, ,\ \ \pd(v)=(\beta(v,0))^L\, ,\ \ \pt(v)=(\beta(v,0))^L\, ,
\ee
in the transfer matrix case and by
\be\label{pqtm}
\pu(v)=e^{\beta h_1}(\alpha_1(v,-u)\beta(u,v))^{N/2}\, ,\ \pd(v)=e^{\beta h_2}(\beta(v,-u)\beta(u,v))^{N/2}\, ,
\ \pt(v)=e^{\beta h_3} (\beta(v,-u)\alpha_3(u,v))^{N/2}\, ,
\ee
in the QTM case.  The eigenvalues of the generalized model are
\cite{Sm,ACDFR,BR}
\be
\Lambda(v)=\pu(\la)\prod_{j=1}^n g_1(\lu_j,v)+\pd(\la)\prod_{i=1}^n g_2(v,\lu_i)\prod_{j=1}^m g_2(\ld_j,v)
+\pt(v)\prod_{j=1}^m g_{3}(v,\ld_j)\, ,
\ee
with $\{\lu_i\}_{i=1}^n\, , \{\ld_j\}_{j=1}^m$ satisfying the BAEs ($g_i(v,w)=\alpha_i(v,w)/\beta(v,w)$)
\begin{align}
&\frac{\pu(\lu_k)}{\pd(\lu_k)}=\prod_{\substack{i=1\\i\ne k}}^n\frac{g_2(\lu_k,\lu_i)}{g_1(\lu_i,\lu_k)}\prod_{j=1}^m g_2(\ld_j,\lu_k)\, ,\ \ \
k=1,\cdots,n\, ,\\
&\frac{\pd(\ld_l)}{\pt(\ld_l)}=\prod_{i=1}^ng^{-1}_2(\ld_l,\lu_i)\prod_{\substack{j=1\\j\ne l}}^m\frac{g_3(\ld_l,\ld_j)}{g_2(\ld_j,\ld_l)}\, ,
\ \ \ \  l=1,\cdots, m\, .
\end{align}
In the QTM case, it is preferable to work with a different pseudovacuum (see
\cite{Sm}) which amounts to cyclic permutations
$(h_1,h_2,h_3)\rightarrow(h_3,h_1,h_2),$
$(\varepsilon_1,\varepsilon_2,\varepsilon_3),
\rightarrow(\varepsilon_3,\varepsilon_1,\varepsilon_2)$ of the chemical
potentials and grading in (\ref{pqtm}).  This explains the order of the
chemical potentials in Eqs.~(\ref{int5}).

\section{Proof of some identities}\label{ap22}

In this appendix we prove several identities which we use in Section
\ref{IntegralE}.  First, we will show that
\be\label{i1}
\lam_2(v)+\lam_3(v)=c\  \frac{\phi_-(v)q_2^{(h)}(v)}{q_1(v)}\, ,\ \ \ q_2^{(h)}(v)=\prod_{j=1}^N\sinh(v-v'^{(2)}_j)\, ,
\ee
with $c$ a constant and $\{v'^{(2)}_j\}_{j=1}^N$ the holes corresponding to the equation $\ad(v)=-1$.
Using (\ref{deflam}) we have
\be\label{int8}
\lam_2(v)+\lam_3(v)=\frac{\phi_-(v)}{q_1(v)}\frac{p_2(v)}{q_2(v)}\, ,
\ee
with $p_2(v)=\phi_+(v)q_1(v+i\g)q_2(v-i\g)e^{\beta
  h_1}+\phi_+(v-i\g)q_1(v)q_2(v+i\g)e^{\beta h_2}$. The function $p_2(v)$ has
the following properties: $p_2(v+i\pi)=(-1)^{3N/2}p_2(v)$ and
$\lim_{v\rightarrow \infty}p_2(v)/(\sinh v)^{3N/2}=const.$ In addition the
equation $p_2(v)=0$ is equivalent with $\ad(v)=-1$, which means that the zeros
of $p_2(v)$ are the $N/2$ Bethe roots and the $N$ holes. This shows that
\[
p_2(v)=c\prod_{j=1}^{N/2}\sinh(v-v^{(2)}_j)\prod_{j=1}^N\sinh(v-v'^{(2)}_j)= c\  q_2(v)q^{(h)}_2(v)\, ,
\]
which together with (\ref{int8}) proves (\ref{i1}). A similar reasoning can be
applied to show
\be\label{i1b}
\lam_1(v)+\lam_2(v)=c\  \frac{\phi_+(v)q_1^{(h)}(v)}{q_2(v)}\, ,\ \ \ q_1^{(h)}(v)=\prod_{j=1}^N\sinh(v-v'^{(1)}_j)\, ,
\ee
where $\{v'^{(1)}_j\}_{j=1}^N$ are the holes corresponding to the equation $\au(v)=-1$.

Consider again the regime $\g\in(0,\pi/2)$ for which the distribution of roots
and holes is presented in Fig. \ref{LEfig}. The second useful result which we
will prove is
\be\label{integralr}
\int_{\mathcal{C}+\mathcal{C}'}d(v-w)\frac{\mathfrak{a}_j'(w)}{1+\mathfrak{a}_j(w)}dw=0\, ,\ \ \ \
d(v-w)=\frac{d}{dv}\ln\sinh(v-w)\, ,
\ee
where $\mathcal{C}'$, see Fig. \ref{LEfig}, is a rectangular contour with the
lower (upper) edges parallel to the real axis through $i(\g-\varepsilon)/2$
and $-i(\g-\varepsilon)/2+i\pi$ with $\varepsilon\rightarrow 0$.
The lower
edge of the contour $\mathcal{C}'$ at $i(\g-\varepsilon)/2$ coincides with the
upper edge of $\mathcal{C}$ but has opposite orientation which means that the
contribution of the $i \pi$-periodic integrand to the integral is zero.
Then, relation (\ref{integralr}) is proved by noticing that the
contributions of the sides parallel to the imaginary axis are also zero as a
result of $\lim_{\Re w\rightarrow\pm\infty}d(v-w)=\mp 1\, $ and
\[
\frac{\mathfrak{a}_j'(w)}{1+\mathfrak{a}_j(w)}=\frac{\mathfrak{a}_j'(\m)}{\mathfrak{a}_j\, (w)}\frac{1}{1+\mathfrak{a}_j^{-1}(w)}\, ,
\]
\[
\lim_{\Re w\rightarrow\pm\infty}\frac{1}{1+\mathfrak{a}_1^{-1}(w)}\rightarrow\frac{1}{1+e^{\beta(h_3-h_1)}}\, ,\ \ \
\lim_{\Re w\rightarrow\pm\infty}\frac{1}{1+\mathfrak{a}_2^{-1}(w)}\rightarrow\frac{1}{1+e^{\beta(h_2-h_1)}}\, ,\ \ \
\lim_{\Re w\rightarrow\pm\infty}\frac{\mathfrak{a}'_j(w)}{\mathfrak{a}_j\, (w)}=0\, .
\]

Using (\ref{i1}) and (\ref{i1b}) and the definition of the auxiliary functions
(\ref{defaux}) we can also easily derive the following identities
\begin{align}
&-\ln \phi_-(v)+\ln q_1(v)-\ln q_1(v+i\g)-\ln q_2(v-i\g)+\ln q_1^{(h)}(v)-\ln(1+\mathfrak{a}_1(v))+c_1=0\, ,\label{zz1}\\
&-\ln \phi_+(v)+\ln q_2(v)-\ln q_1(v+i\g)-\ln q_2(v-i\g)+\ln q_2^{(h)}(v)-\ln(1+\mathfrak{a}_2(v))+c_2=0\, ,\label{zz2}
\end{align}
with $c_{1,2}$ constants and arbitrary $v$.

\newpage

\setcounter{equation}{0}%
\setcounter{page}{1}
\numberwithin{equation}{section}
\setcounter{section}{0}
\renewcommand{\thesection}{\arabic{section}}

\begin{center}
{\Large Supplemental Material for EPAPS \\
Thermodynamics, density profiles and correlation functions of the inhomogeneous one-dimensional spinor Bose gas }
\end{center}
\section{Algebraic Bethe ansatz for the generalized $q=3$ Perk-Schultz model}

Here we solve the generalized $q=3$ Perk-Schultz
model using the method developed by G\"ohmann \cite{G1,GS} for the $gl(1|2)$
generalized model. The eigenvalues of the transfer matrix (\ref{transferm})
and quantum transfer matrix (\ref{qtransferm}) can be obtained as particular
cases of this general solution.

The generalized $q=3$ Perk-Schultz model is the set of all linear
representations of the Yang-Baxter algebra
\be\label{int16}
\check\R(v,w)[\T(v)\otimes\T(w)]=[\T(v)\otimes\T(w)]\check\R(v,w)\, ,
\ee
with the $\R$-matrix defined in (\ref{rm}) and $\T(v)$ the monodromy matrix
which acts triangularly on a pseudovacuum (highest vector) $|\Omega\rangle$
\be
\T(v)|\Omega\rangle=\left(\begin{array}{ccc} A(v)|\Omega\rangle & B_1(v)|\Omega\rangle & B_2(v)|\Omega\rangle \\
                                             C_1(v)|\Omega\rangle & D_{11}(v)|\Omega\rangle & D_{12}(v)|\Omega\rangle\\
                                             C_2(v)|\Omega\rangle & D_{21}(v)|\Omega\rangle & D_{22}(v)|\Omega\rangle\\
                       \end{array}\right)
                    =\left(\begin{array}{ccr} \pu(v)|\Omega\rangle & B_1(v)|\Omega\rangle & B_2(v)|\Omega\rangle \\
                                              0 & \pd(v)|\Omega\rangle & D_{12}(v)|\Omega\rangle\\
                                              0 &  0                  &\pt(v)|\Omega\rangle
                    \end{array}\right)\, .
\ee
The functions $\varphi_j(v)\, ,j=1,2,3$ are called the parameters of the
model. The monodromy matrices of the transfer matrix (\ref{int3}) and QTM
(\ref{mq}) satisfy these requirements with pseudovacua (\ref{vacuumt}) and
(\ref{pseudoQTM}) respectively. The parameters of the model are given by
\be\label{parameterst}
\pu(v)=(\alpha_1(v,0))^L\, ,\ \ \pd(v)=(\beta(v,0))^L\, ,\ \ \pt(v)=(\beta(v,0))^L\, ,
\ee
in the transfer matrix case and by
\be\label{parametersq}
\pu(v)=e^{\beta h_1}(\alpha_1(v,-u)\beta(u,v))^{N/2}\, ,\ \pd(v)=e^{\beta h_2}(\beta(v,-u)\beta(u,v))^{N/2}\, ,
\ \pt(v)=e^{\beta h_3} (\beta(v,-u)\alpha_3(u,v))^{N/2}\, ,
\ee
in the QTM case. We will obtain the eigenvalues of
\be
\tr (v)=\mbox{tr}_0\T(v)=A(v)+D_{11}(v)+D_{22}(v)\, ,
\ee
using only the Yang-Baxter algebra and the fact that the monodromy matrix acts
triangularly on the pseudovacuum. It is also important to notice that we do
not require that $D_{12}(v)|\Omega\rangle=0$.

Before we start we need to obtain some auxiliary results which will play an
important role in the following considerations. For $Q$ integer and
$\{v_j\}_{j=1}^Q$ a set of inhomogeneities we consider the following monodromy
matrix
\be
\bar\T(v,\{v_j\})=\Lo_Q(v,v_Q)\Lo_{Q-1}(v,v_{Q-1})\cdots\Lo_1(v,v_1)\, ,\ \ \ \ \ \
\ee
with the $\Lo$-operators defined as
\be\label{int17}
\Lo_j(v,v_j)=\sum_{a,b,a_1,b_1=1}^3\R_{b\, b_1}^{aa_1}(v, v_j)e_{ab}^{(0)}e_{a_1b_1}^{(j)}\, ,
\ \ \ \ \  \Lo_j(\la,v_j)\in \mbox{End}\left((\mathbb{C}^3)^{\otimes(Q+1)}\right)\, ,
\ee
where $e_{ab}^{(j)}$ is the canonical basis of operators acting on
$(\mathbb{C}^3)^{\otimes(Q+1)}$. Using the formula
$e_{ab}^{(j)}e_{cd}^{(j)}=\delta_{bc}e_{ad}^{(j)}$ and the definition
(\ref{int17}) we can obtain explicitly the elements of the monodromy matrix
\[
\bar \T(v,\{v_j\})=\bar \T_{b\,b_Q\cdots\, b_1}^{aa_Q\cdots a_1}(v,\{v_j\})e_{ab}^{(0)}e_{a_Qb_Q}^{(Q)}\cdots e_{a_1b_1}^{(1)}\, ,
\]
as
\be\label{elements}
\bar \T_{b\,b_Q\cdots\, b_1}^{aa_Q\cdots a_1}(v,\{v_j\})=\sum_{c_Q,\cdots,c_2=1}^3\R_{c_Qb_Q}^{a\ \,a_Q}(v,v_Q)\R_{c_{Q-1}b_{Q-1}}^{c_Q\ \ \, a_{Q-1}}(v,v_{Q-1})\cdots\R_{c_2b_2}^{c_3a_2}(v,v_2)\R_{b\ \, b_1}^{c_2 a_1}(v,v_1)\, .
\ee
Then, the elements of the associated transfer matrix $\bar \tr(v,\{v_j\})=\mbox{tr}_0\bar\T(v,\{v_j\})$ are
\be\label{int18}
\bar \tr_{b_Q\cdots\, b_1}^{a_Q\cdots a_1}(v,\{v_j\})=\sum_{a,c_Q,\cdots,c_2=1}^3\R_{c_Qb_Q}^{a\ \,a_Q}(v,v_Q)\R_{c_{Q-1}b_{Q-1}}^{c_Q\ \ \, a_{Q-1}}(v,v_{Q-1})\cdots\R_{c_2b_2}^{c_3a_2}(v,v_2)\R_{a\ \, b_1}^{c_2 a_1}(v,v_1)\, .
\ee
The previous formulas can be easily modified to the case of different
$\R$-matrices and number of lattice sites $Q$ and, will be used extensively,
especially (\ref{int18}), in the following computations. Now we are ready to
solve the generalized model.

\subsection{Bethe ansatz  for the first level}\label{ap1}

The Yang-Baxter algebra (\ref{int16}) can be rewritten in a more explicit form
if we use the following block form for the monodromy matrix
\[
\T(v)=\left(\begin{array}{cc}A(v)&B(v)\\
                                      C(v)&D(v)
                     \end{array}\right),\ \ \
B(v)=(B_1(v)\  B_2(v))\, , \ \ C(\la)=\left(\begin{array}{c}C_1(v)\\C_2(v)\end{array}\right)\, , \ \
D(v)=\left(\begin{array}{cc}D_{11}(v)&D_{12}(v)\\
                                      D_{21}(v)&D_{22}(v)
                     \end{array}\right)\, .
\]
Introducing the matrix
\be\label{defx}
X=\left(\begin{array}{ccccc}
           \mathbb{I}_4& & & & \\
            & 0 & 0 & 1 & \\
            & 1 & 0 & 0 & \\
            & 0 & 1 & 0 & \\
            &   &   &   & \mathbb{I}_2
          \end{array}\right)\, ,
\ee
which has the properties: permutes circularly the 5th, 6th and 7th row when
multiplied from the left, permutes circularly the 5th, 6th and 7th row when
multiplied from the right and $XX^T=X^TX=\mathbb{I}_9$, we can perform the
following similarity transformation on the Yang-Baxter algebra (\ref{int16})
\be
(X\check\R(v,w)X^T)[X(\T(v)\otimes\T(w))X^T]=[X(\T(v)\otimes\T(w))X^T](X\check\R(v,w)X^T)\, ,
\ee
which can be written as
\begin{align}\label{int19}
&\left(\begin{array}{cccc}
         \alpha_1(v,w)&0&0&0\\
         0&\g_+(v,w)\mathbb{I}_2&\beta(v,w)\mathbb{I}_2&0\\
         0&\beta(v,w)\mathbb{I}_2&\g_-(v,w)\mathbb{I}_2&0\\
         0&0&0&\check{\bar{\R}}^{(1)}(v,w)
       \end{array}\right)
\left(\begin{array}{cccc}
        A(\la)\otimes A(\m) & A(\la)\otimes B(\m) &B(\la)\otimes A(\m) & B(\la)\otimes B(\m)\\
        A(\la)\otimes C(\m) & A(\la)\otimes D(\m) &B(\la)\otimes C(\m) & B(\la)\otimes D(\m)\\
        C(\la)\otimes A(\m) & C(\la)\otimes B(\m) &D(\la)\otimes A(\m) & D(\la)\otimes B(\m)\\
        C(\la)\otimes C(\m) & C(\la)\otimes D(\m) &D(\la)\otimes C(\m) & D(\la)\otimes D(\m)\\
      \end{array}\right)\nonumber\\
&= \left(\begin{array}{cccc}
        A(\m)\otimes A(\la) & A(\m)\otimes B(\la) &B(\m)\otimes A(\la) & B(\m)\otimes B(\la)\\
        A(\m)\otimes C(\la) & A(\m)\otimes D(\la) &B(\m)\otimes C(\la) & B(\m)\otimes D(\la)\\
        C(\m)\otimes A(\la) & C(\m)\otimes B(\la) &D(\m)\otimes A(\la) & D(\m)\otimes B(\la)\\
        C(\m)\otimes C(\la) & C(\m)\otimes D(\la) &D(\m)\otimes C(\la) & D(\m)\otimes D(\la)\\
      \end{array}\right)
\left(\begin{array}{cccc}
         \alpha(v,w)&0&0&0\\
         0&\g_+(v,w)\mathbb{I}_2&\beta(v,w)\mathbb{I}_2&0\\
         0&\beta(v,w)\mathbb{I}_2&\g_-(v,w)\mathbb{I}_2&0\\
         0&0&0&\check{\bar{\R}}^{(1)}(v,w)
       \end{array}\right)
\end{align}
where $\check{\bar{\R}}^{(1)}(v,w)$ is the $\check\R$-matrix of the
XXZ spin-chain (modulo a gauge transformation)
\be\label{rmatrix1}
\check{\bar{\R}}^{(1)}(v,w)=\left(\begin{array}{cccc}
                          \alpha_2(v,w)&0&0&0\\
                          0&\g_+(v,w)&\beta(v,w)&0\\
                          0&\beta(v,w)&\g_-(v,w)&0\\
                          0&0&0&\alpha_3(v,w)
                        \end{array}\right)\, .
\ee
For the application of algebraic Bethe ansatz we need the following
commutation relations which can be extracted from (\ref{int19})
\begin{subequations}
\begin{align}
A(v)A(w)&=A(w)A(v)\, ,\label{cr1}\\
B_{a_1}(v)B_{a_2}(w)&=\check\R^{(1)}\,_{a_1a_2}^{b_1b_2}(v,w)B_{b_1}(w)B_{b_2}(v)\, ,\label{cr2}\\
A(v)B_b(w)&=g_1(w,v)B_b(w)A(v)-h_-(w,v)B_b(v)A(w)\, ,\label{cr3}\\
D_{ab_1}(v)B_{b_2}(w)&=g_1(v,w)\check\R^{(1)}\,_{b_1b_2}^{c_1c_2}(v,w)B_{c_1}(w)D_{ac_2}(v)-h_+(v,w)B_{b_1}(v)D_{ab_2}(w)\label{cr4}\, ,
\end{align}
\end{subequations}
where
\be
\check\R^{(1)}(v,w)=\frac{1}{\alpha_1(v,w)}\check{\bar{\R}}^{(1)}(v,w)\, ,\ \
g_i(v,w)=\frac{\alpha_i(v,w)}{\beta(v,w)}\, ,\ \
h_{\pm}(v,w)=\frac{\g_\pm(v,w)}{\beta(v,w)}\, .
\ee
Our goal is to solve the eigenvalue problem
\be\label{ee}
\tr(v)|\Phi\rangle=\Lambda(v)|\Phi\rangle\, ,\ \  |\Phi\rangle\in\mathcal{H}\, ,
\ee
with $\tr(v)=A(v)+D_{11}(v)+D_{22}(v)$ and $\mathcal{H}$ the Hilbert space of the
model. We are going to look for eigenvectors of the form
\[
|\Phi\rangle=\sum_{a_1,\cdots,a_n=1}^2F^{a_n\cdots a_1}B_{a_1}(\lu_1)\cdots B_{a_n}(\lu_n)=
   B(\lu_1)\otimes\cdots\otimes B(\lu_n)F\, ,
\]
where $B(\lu_1)\otimes\cdots\otimes B(\lu_n)=\bigotimes_{j=1}^nB(\lu_j)$ can
be viewed as a $2^n$ component row vector with entries acting on
$\mathcal{H}$, $F$ as a column vector with $2^n$ components belonging to the
Hilbert space $\mathcal{H}$, and $\{\lu_j\}_{j=1}^n$ are parameters which will
satisfy Bethe equations.  A consequence of commutation relation (\ref{cr2}) is
the following:
\begin{lem}\label{lemma}
Let $\tr^{(0)}(v,\{\lu_j\})$ be the transfer matrix of a system of $n$ lattice
sites, inhomogeneities $\{\lu_j\}$ and $\R$-matrix
$\R^{(1)}(v,w)=\Pe\check\R^{(1)}(v,w)$. Then
\be\label{int20}
B(\lu_1)\otimes\cdots\otimes B(\lu_n)=B(\lu_2)\otimes B(\lu_3)\otimes\cdots\otimes B(\lu_n)\otimes B(\lu_1)\tr^{(n)}(\lu_1,\{\lu_j\})\, .
\ee
\end{lem}
\begin{proof}
Starting with $B_{b_1}(\lu_1)B_{b_2}(\lu_2)\cdots B_{b_n}(\lu_n)$ and applying successively (\ref{cr2}) we find
\begin{align*}
&B_{b_1}(\lu_1)B_{b_2}(\lu_2)\cdots B_{b_n}(\lu_n)\\
&\ \ \ \ \ \ \ =\check\R^{(1)}\,_{c_n\, b_n}^{a_na_1}(\lu_1,\lu_n)
\check\R^{(1)}\,_{c_{n-1} b_{n-1}}^{a_{n-1} c_{n}}(\lu_1,\lu_{n-1})\cdots\check\R^{(1)}\,_{b_1\, b_2}^{a_2 c_3}(\lu_1,\lu_2)
B_{a_2}(\lu_2)B_{a_3}(\lu_3)\cdots B_{a_n}(\lu_n)B_{a_1}(\lu_1)\, ,\\
&\ \ \ \ \ \ \ =\R^{(1)}\,_{c_n\, b_n}^{a_1a_n}(\lu_1,\lu_n)
\R^{(1)}\,_{c_{n-1} b_{n-1}}^{c_{n}\ \ a_{n-1} }(\lu_1,\lu_{n-1})\cdots\R^{(1)}\,_{b_1\, b_2}^{c_3 a_2}(\lu_1,\lu_2)
B_{a_2}(\lu_2)B_{a_3}(\lu_3)\cdots B_{a_n}(\lu_n)B_{a_1}(\lu_1)\, ,
\end{align*}
Comparing with (\ref{int18}) modified for the $\R$-matrix $\R^{(1)}(v,w)$ and
using $\check\R^{(1)}\,_{a\ b_1}^{c_2
  a_1}(\lu_1,\lu_1)=\delta_{c_2b_1}\delta_{aa_1}$ we see that
\[
B_{b_1}(\lu_1)B_{b_2}(\lu_2)\cdots B_{b_n}(\lu_n)=\tr^{(n)}\,_{b_n\cdots b_1}^{a_n\cdots a_1}(\lu_1,\{\lu_j\})
B_{a_2}(\lu_2)B_{a_3}(\lu_3)\cdots B_{a_n}(\lu_n)B_{a_1}(\lu_1)
\]
proving (\ref{int20}).
\end{proof}
Lemma \ref{lemma} also shows that the cyclic permutation $B(\lu_i)\rightarrow
B(\lu_{i+1})$ followed by the multiplication with the matrix
$U=\tr^{(0)}(\lu_1,\{\lu_i\})\, ,$ leaves the product $B(\lu_1)\otimes
\cdots\otimes B(\lu_n)$ invariant.

Following the general strategy of the algebraic Bethe ansatz we are going to
act with $A(v)+D_{11}(v)+D_{22}(v)$ on the eigenstate $|\Phi\rangle$ and then
use the commutation relations (\ref{cr3}) and (\ref{cr4}) in order to move
this operators to the right.  Due to the fact that in the r.h.s of the
commutation relations we have two terms, moving these operators to the right
will produce $2^n$ terms which we will divide in two categories: wanted and
unwanted. The wanted terms are those which contain the product
$B(\lu_1)\otimes\cdots\otimes B(\lu_n)$ and the unwanted ones have one of the
$B(\lu_i)$ replaced with $B(\la).$ In keeping track of the unwanted terms the
symmetry described in Lemma \ref{lemma} will play an important
role. Satisfying the eigenvalue equation (\ref{ee}) requires that the unwanted
terms should vanish, the condition which will produce the Bethe equations.

{\it Action of $A(v)$ on $[\bigotimes_{j=1}^nB(\lu_j)].$} The wanted term is
produced by using only the first term from the r.h.s of (\ref{cr3}) in moving
$A(v)$ to the right obtaining
\be\label{int21}
[\bigotimes_{j=1}^nB(\lu_j)]A(v)\prod_{j=1}^ng_1(\lu_j,v)\, .
\ee
The first unwanted term is obtained by using the second term of the
commutation relation (\ref{cr3}) in order to move past $B(\lu_1)$ and then
using the first term $n-1$ times with the result
\[
-[B(v)\otimes B(\lu_2)\otimes\cdots\otimes B(\lu_n)]A(\lu_1)h_-(\lu_1,v)\prod_{j=2}^n g_1(\lu_j,\lu_1)\, .
\]
Making use of the $U$ transformation described in Lemma \ref{lemma} we can
derive a general formula for the $k$-th unwanted term
\be\label{int22}
-[B(v)\otimes B(\lu_{k+1})\otimes\cdots\otimes
B(\lu_n)\otimes B(\lu_1)\otimes\cdots\otimes B(\lu_{k-1})]U^{k-1}A(\lu_k)h_-(\lu_k,v)\prod_{\substack{j=1\\j\ne k}}^n g_1(\lu_j,\lu_k)\, .
\ee
Adding (\ref{int21}) and (\ref{int22}) and introducing the notation
\be
S^{k-1}[\bigotimes_{j=1}^nB(\lu_j)]=B(v)\otimes B(\lu_{k+1})\otimes\cdots\otimes
B(\lu_n)\otimes B(\lu_1)\otimes\cdots\otimes B(\lu_{k-1})\, ,\ \  k=1,\cdots,n
\ee
where $S^0[\bigotimes_{j=1}^nB(\lu_j)]=B(v)\otimes B(\lu_2)\otimes\cdots\otimes B(\lu_n)$
then
\be\label{i8}
A(v)[\bigotimes_{j=1}^nB(\lu_j)]=[\bigotimes_{j=1}^nB(\lu_j)]\prod_{j=1}^ng_1(\lu_j,v)A(v)
-\sum_{k=1}^n(S^{k-1}[\bigotimes_{j=1}^nB(\lu_j)])U^{k-1}h_-(\lu_k,v)\prod_{\substack{j=1\\j\ne k}}^n g_1(\lu_j,\lu_k)A(\lu_k)\, .
\ee
%

{\it Action of $D_{11}(v)+D_{22}(v)$ on $[\bigotimes_{j=1}^nB(\lu_j)].$} We
will start with $D_{11}(v).$ The wanted term is obtained by using only the
first term in the r.h.s of (\ref{cr4}) when moving to the right with the
result
\begin{align}\label{int23}
&D_{11}(\la)B_{b_1}(\lu_1)B_{b_2}(\lu_2)\cdots B_{b_n}(\lu_n)\nonumber\\
&=\prod_{j=1}^ng_1(v,\lu_j)\check\R^{(1)}\,_{c_n b_n}^{a_n c}(v,\lu_n)\cdots
\check\R^{(1)}\,_{c_2\, b_2}^{a_2 c_3}(v,\lu_2)\check\R^{(1)}\,_{d_1\, b_1}^{a_1 c_2}(v,\lu_1)
 B_{a_1}(\lu_1)\cdots B_{a_n}(\lu_n)\delta_{1d_1}D_{1c}(v)\, ,\nonumber\\
&=\prod_{j=1}^ng_1(v,\lu_j)\R^{(1)}\,_{c_n b_n}^{c\  a_n }(v,\lu_n)\cdots
\R^{(1)}\,_{c_2\, b_2}^{c_3 a_2}(v,\lu_2)\R^{(1)}\,_{d_1 b_1}^{c_2 a_1}(v,\lu_1)
B_{a_1}(\lu_1)\cdots B_{a_n}(\lu_n)\delta_{1d_1}D_{1c}(v)\, ,\nonumber\\
&=\prod_{j=1}^ng_1(v,\lu_j)B_{a_1}(\lu_1)\cdots B_{a_n}(\lu_n)D_{1c}(v)\T^{(0)}\,_{1b_n\cdots\, b_1}^{ca_n\cdots a_1}(v,\{\lu_j\})
\end{align}
The last line of (\ref{int23}) was derived using (\ref{elements}) with
$\T^{(0)}(v,\{\lu_j\})$ the monodromy matrix of a system with $n$ lattice
sites $\R$-matrix $\R^{(1)}(v,w)$ and inhomogeneites $\{\lu_j\}.$ Introducing
\be\label{deft1}
D^{(0)}(\la)=D(v)\otimes \mathbb{I}_2^{\otimes n}\, ,\ \ \ \text{and} \ \ \ \T^{(1)}(v,\{\lu_i\})=D^{(0)}(v)\T^{(0)}(v,\{\lu_i\})\, ,
\ee
then $D_{1c}(v)\T^{(0)}\,_{1b_n\cdots\, b_1}^{ca_n\cdots
  a_1}(v,\{\lu_j\})=\T^{(1)}_{11}(v,\{\lu_j\})$ where we have denoted by
$\T^{(1)}_{ab}(v,\{\lu_j\})$ the elements of the monodromy matrix
\ $\T^{(1)}(v,\{\lu_j\})$ in the auxiliary space. The computations for
$D_{22}(v)$ are similar but in this case the result involves
$\T^{(1)}_{22}(v,\{\lu_j\})$. Therefore the wanted term is
\be\label{i7}
(D_{11}(v)+D_{22}(v))[\bigotimes_{j=1}^nB(\lu_j)]=
[\bigotimes_{j=1}^nB(\lu_j)]\left(\T^{(1)}_{11}(v,\{\lu_j\})+\T^{(1)}_{22}(v,\{\lu_j\})\right)\prod_{j=1}^ng_1(v,\lu_j)\, .
\ee
The first unwanted term is obtained using once the second term on the r.h.s of
(\ref{cr4}) and $n-1$ times the first term with the result
\begin{align}\label{int24}
&D_{11}(\la)B_{b_1}(\lu_1)B_{b_2}(\lu_2)\cdots B_{b_n}(\lu_n)\nonumber\\
&=-h_+(v,\lu_1)\prod_{j=2}^ng_1(\lu_1,\lu_j)\check\R^{(1)}\,_{c_n b_n}^{a_n c}(\lu_1,\lu_n)\cdots
\check\R^{(1)}\,_{c_3\, b_3}^{a_3 c_4}(\lu_1,\lu_3)\check\R^{(1)}\,_{b_1\, b_2}^{a_2 c_3}(\lu_1,\lu_2)
 B_{a_1}(\lu_1)\cdots B_{a_n}(\lu_n)\delta_{1a_1}D_{1c}(v)\, ,\nonumber\\
&=-h_+(v,\lu_1)\prod_{j=2}^ng_1(\lu_1,\lu_j)\R^{(1)}\,_{c_n b_n}^{c a_n }(\lu_1,\lu_n)\cdots
\R^{(1)}\,_{c_3\, b_3}^{c_4 a_3}(\lu_1,\lu_3)\R^{(1)}\,_{b_1\, b_2}^{c_3 a_2}(\lu_1,\lu_2)
 B_{a_1}(\lu_1)\cdots B_{a_n}(\lu_n)\delta_{1a_1}D_{1c}(v)\, ,\nonumber\\
 &=-h_+(v,\lu_1)\prod_{j=2}^ng_1(\lu_1,\lu_j)B_{a_1}(\lu_1)\cdots B_{a_n}(\lu_n)D_{1c}(v)\T^{(0)}\,_{1b_n\cdots\, b_1}^{ca_n\cdots a_1}(\lu_1,\{\lu_j\})
\end{align}
where the last line was obtained using (\ref{elements}) and $\R^{(1)}\,_{a\ \,
  b_1}^{c_2 a_1}(\lu_1,\lu_1)=\delta_{c_2b_1}\delta_{aa_1}$.  Using the
definition of the monodromy matrix (\ref{deft1}) the first unwanted term for
$D_{11}(v)$ is
\[
-S^0[\bigotimes_{j=1}^nB(\lu_j)]\T^{(1)}_{11}(\lu_1,\{\lu_j\})h_+(v,\lu_1)\prod_{j=2}^ng_1(\lu_1,\lu_j)
\]
and a similar expression involving $\T^{(1)}_{22}(\lu_1,\{\lu_j\})$ for
$D_{22}(v)$. The general form of the unwanted term is
\be\label{int25}
-S^{k-1}[\bigotimes_{j=1}^nB(\lu_j)]U^{k-1}\left(\T^{(1)}_{11}(\lu_k,\{\lu_j\})+
\T^{(1)}_{22}(\lu_k,\{\lu_j\})\right)h_+(v,\lu_k)\prod_{\substack{j=1\\j\ne k}}^ng_1(\lu_k,\lu_j)
\ee
Finally, collecting (\ref{i8}), (\ref{i7}) and (\ref{int25}) we find
\begin{align}\label{ii8}
&(A(v)+D_{11}(v)+D_{22}(v))[\bigotimes_{j=1}^nB(\lu_j)]=[\bigotimes_{j=1}^nB(\lu_j)]\left(\prod_{j=1}^ng_1(\lu_j,v)A(v)+
\prod_{j=1}^ng_1(v,\lu_j)\mbox{tr}\T^{(1)}(v,\{\lu_j\})\right)\\
&-\sum_{k=1}^n(S^{k-1}[\bigotimes_{j=1}^nB(\lu_j)])U^{k-1}
\left( h_-(\lu_k,v)\prod_{\substack{j=1\\j\ne k}}^n g_1(\lu_j,\lu_k)A(\lu_k)+h_+(v,\lu_k)\prod_{\substack{j=1\\j\ne k}}^ng_1(\lu_k,\lu_j)\mbox{tr}\T^{(1)}(\lu_k,\{\lu_j\})\right)\, .\nonumber
\end{align}
%

\subsection{Bethe ansatz for the second level}\label{ap2}

If we introduce the "vacuum subspace" \cite{KR1,R,G1},
$\mathcal{H}_0\in\mathcal{H}$ characterized by the conditions
\be
A(v)|\Phi\rangle=\pu(v)|\Phi\rangle\, ,\ \ \ C(v)|\Phi\rangle=0\, ,\ \ |\Phi\rangle\in\mathcal{H}_0\, ,
\ee
then we can prove the following
\begin{lem}
$\mathcal{H}_0$ is invariant under the action of $D(v)$.
\end{lem}
\begin{proof}
The following relations
\begin{align}
\g_+(v,w)B(v)\otimes C(w)+\beta(v,w)D(w)\otimes A(w)&=\beta(v,w)A(w)\otimes D(v)+\g_-(v,w)B(w)\otimes C(v)\, ,\\
\check{\bar\R}^{(1)}(v,w)D(v)\otimes C(w)&=\beta(v,w)C(w)\otimes D(v)+\g_-(v,w)D(w)\otimes C(v)\, ,
\end{align}
can be obtained from the block form of the Yang-Baxter algebra
(\ref{int19}). Using these relations we have
\[
A(w)\otimes D(v)|\Phi\rangle=\pu(w)D(v)|\Phi\rangle\, ,\ \ C(w)\otimes D(v)|\Phi\rangle=0\, ,
\]
for all $|\Phi\rangle\in\mathcal{H}_0$ finishing the proof.
\end{proof}
A corollary of this lemma is that the linear space spanned by all linear
combination of vectors of the form $D_{12}(\ld_1)\cdots
D_{12}(\ld_m)|\Omega\rangle$ is a linear subspace of $\mathcal{H}_0$.

Using  the definition (\ref{deft1}) and the notation
\[
 \T^{(0)}(\la,\{\lu_j\})=\left(\begin{array}{cc}
                        A^{(0)}(\la,\{\lu_i\})&B^{(0)}(\la,\{\lu_i\})\\
                        C^{(0)}(\la,\{\lu_i\})&D^{(0)}(\la,\{\lu_i\})
                   \end{array}\right)\, ,
\]
the elements of $\T^{(1)}(\la,\{\lu_j\})$ in the auxiliary space are
\begin{align}\label{i10}
\T^{(1)}(\la,\{\lu_j\})=&\left(\begin{array}{cc}
                        A^{(1)}(\la,\{\lu_j\})&B^{(1)}(\la,\{\lu_j\})\\
                        C^{(1)}(\la,\{\lu_j\})&D^{(1)}(\la,\{\lu_j\})
                   \end{array}\right)\, ,\\
             =&\left(\begin{array}{cc}
                     A^{(0)}(\la,\{\lu_j\})D_{11}(\la)+C^{(0)}(\la,\{\lu_j\})D_{12}(\la)& B^{(0)}(\la,\{\lu_j\})D_{11}(\la)+D^{(0)}(\la,\{\lu_j\})D_{12}(\la)\\
                     A^{(0)}(\la,\{\lu_j\})D_{12}(\la)+C^{(0)}(\la,\{\lu_j\})D_{22}(\la)& B^{(0)}(\la,\{\lu_j\})D_{12}(\la)+D^{(0)}(\la,\{\lu_j\})D_{22}(\la)
                     \end{array}\right)\, ,\nonumber
\end{align}
where we have used $[D_{ab}(v),\T^{(0)}_{cd}(v,\{\lu_j\})]=0.$ $D(v)$ (see the
element (4,4) of (\ref{int19})) and $\T^{(0)}(v,\{\lu_j\})$ are both
representations of the Yang-Baxter algebra with the same $\R$-matrix
$\R^{(1)}$, which means that
\be\label{yb1}
\check\R^{(1)}(v,w)[\T^{(1)}(v)\otimes\T^{(1)}(w)]=[\T^{(1)}(w)\otimes\T^{(1)}(v)]\check\R^{(1)}(v,w)\, .
\ee
It follows that $\mbox{tr}\T^{(1)}(v,\{\lu_j\})$ can be diagonalized using ABA
if we can find a pseudovacuum on which $\T^{(1)}(v,\{\lu_j\})$ acts
triangularly. $\T^{(1)}(v,\{\lu_j\})$ acts on $(\mathbb{C}^2)^{\otimes
  n}\otimes\mathcal{H}$ and, as we will show below,
\be
|\bar\Omega\rangle=|\Omega^{(0)}\rangle\otimes|\Omega\rangle=
\left(\begin{array}{c}1\\0\end{array}\right)^{\otimes n}\otimes|\Omega\rangle\, ,
\ee
is an appropriate pseudovacuum. The monodromy matrix $\T^{(0)}(v,\{\lu_i\})$
is the ordered product of the following $\Lo$-operators
\be
\bar\Lo_j(v,\lu_j)=\sum_{a,b,a_1,b_1=1}^2\R^{(1)}\,_{b\, b_1}^{aa_1}(v, \lu_j)e_{ab}^{(0)}e_{a_1b_1}^{(j)}\, ,
\ \ \ \ \  \bar\Lo_j(v,\lu_j)\in \mbox{End}\left((\mathbb{C}^2)^{\otimes(n+1)}\right)\, ,
\ee
which can be represented in the auxiliary space as
\be\label{int26}
\bar\Lo_j(v,\lu_j)=\left(\begin{array}{lr}
                                     \frac{\alpha_2(v,v_j^{(1)})}{\alpha_1(v,v_j^{(1)})}e_{11}^{(j)}+g^{-1}_1(v,\lu_j)e_{22}^{(j)} & \frac{\g_+(v,\lu_j)}{\alpha_1(v,\lu_j)}e_{21}^{(j)}\\
                                      \frac{\g_-(v,\lu_j)}{\alpha_1(v,\lu_j)}e_{12}^{(j)} &  g^{-1}_1(v,\lu_j)e_{11}^{(j)}+\frac{\alpha_3(v,v_j^{(1)})}{\alpha_1(v,v_j^{(1)})}e_{22}^{(j)}
                                   \end{array}\right)\, ,
\ee
with $e^{(j)}_{ab}$ a canonical basis of operators acting on  $(\mathbb{C}^2)^{\otimes n}$. Using
the representation (\ref{int26}) we have
\be
\T^{(0)}(v,\{\lu_j\})|\Omega^{(0)}\rangle =\left(\begin{array}{lr}
                                     \prod_{j=1}^n\frac{\alpha_2(v,v_j^{(1)})}{\alpha_1(v,v_j^{(1)})}|\Omega^{(0)}\rangle  & B^{(0)}(v,\{\lu_j\})|\Omega^{(0)}\rangle \\
                                      0 &  \prod_{j=1}^ng^{-1}_1(v,\lu_j)|\Omega^{(0)}\rangle
                                   \end{array}\right)\, ,
\ee
and as a consequence of (\ref{i10})
\be
\T^{(1)}(v,\{\lu_j\})|\bar\Omega\rangle =\left(\begin{array}{lr}
                                     \pd(v)\prod_{j=1}^n\frac{\alpha_2(v,v_j^{(1)})}{\alpha_1(v,v_j^{(1)})}|\bar\Omega\rangle  & B^{(1)}(v,\{\lu_j\})|\bar\Omega\rangle \\
                                      0 &  \pt(v)\prod_{j=1}^ng^{-1}_1(v,\lu_j)|\bar\Omega\rangle
                                   \end{array}\right)\, ,
\ee
which proves that $|\bar\Omega\rangle$ is a proper pseudovacuum.  The
diagonalization of
$\mbox{tr}\T^{(1)}(v,\{\lu_j\})=A^{(1)}(v,\{\lu_j\})+B^{(1)}(v,\{\lu_j\})$ is
very similar to the case of an inhomogeneous XXZ spin chain \cite{KBI}. The
following commutation relations can be extracted from the Yang-Baxter algebra
(\ref{yb1})
\begin{subequations}
\begin{align}
A^{(1)}(\la)A^{(1)}(\m)&=A^{(1)}(\m)A^{(1)}(\la)\, ,\\
B^{(1)}(\la)B^{(1)}(\m)&=B^{(1)}(\m)B^{(1)}(\la)\, ,\\
A^{(1)}(\la)B^{(1)}(\m)&=g_2(\m,\la)B^{(1)}(\m)A^{(1)}(\la)-h_-(\m,\la)B^{(1)}(\la)A^{(1)}(\m)\, ,\label{int27}\\
D^{(1)}(\la)B^{(1)}(\m)&=g_3(\la,\m)B^{(1)}(\m)D^{(1)}(\la)-h_+(\la,\m)B^{(1)}(\la)D^{(1)}(\m)\, .\label{int28}
\end{align}
\end{subequations}
We are looking for eigenvectors of the form
$|\Phi^{(1)}\rangle=B^{(1)}(\ld_1)\cdots B^{(1)}(\ld_m)|\bar\Omega\rangle$
solving the eigenvalue equation
\be\label{int29}
\mbox{tr}\T^{(1)}(v,\{\lu_j\})|\Phi^{(1)}\rangle=\Lambda^{(1)}(v)|\Phi^{(1)}\rangle\, .
\ee
Moving $A^{(1)}(v)$ and $D^{(1)}(v)$ to the right with the help of the
commutation relations (\ref{int27}) and (\ref{int28}) we obtain
\begin{align}\label{fe2}
&(A^{(1)}(v)+D^{(1)}(v))[\prod_{j=1}^m B^{(1)}(\ld_j)]=
[\prod_{j=1}^m B^{(1)}(\ld_j)] \left(\prod_{j=1}^mg_2(\ld_j,\la)A^{(1)}(v)+\prod_{j=1}^m g_3(\la,\ld_j)D^{(1)}(v)\right)\nonumber\\
&-\sum_{k=1}^m[B^{(1)}(\la)\prod_{\substack{j=1\\j\ne k}}^m B^{(1)}(\ld_j)]\left(h_-(\ld_k,v)\prod_{\substack{j=1\\j\ne k}}^m g_2(\ld_j,\ld_k)A^{(1)}(\ld_k)
+h_+(v,\ld_k)\prod_{\substack{j=1\\j\ne k}}^m g_3(\ld_k,\ld_j)D^{(1)}(\ld_k)
\right)\, .
\end{align}
Using
$A^{(1)}(v)|\bar\Omega\rangle=\pd(v)\prod_{j=1}^n\frac{\alpha_2(v,v_j^{(1)})}{\alpha_1(v,v_j^{(1)})}|\bar\Omega\rangle$
and
$D^{(1)}(v)|\bar\Omega\rangle=\pt(v)\prod_{j=1}^ng^{-1}_1(v,\lu_j)|\bar\Omega\rangle$
then $|\Phi^{(1)}\rangle=B^{(1)}(\ld_1)\cdots$ $
B^{(1)}(\ld_m)|\bar\Omega\rangle$ is an eigenvector of
$\mbox{tr}\T^{(1)}(v,\{\lu_j\})$ with eigenvalue
\be
\Lambda^{(1)}(v)=\pd(v)\prod_{j=1}^n\frac{\alpha_2(v,v_j^{(1)})}{\alpha_1(v,v_j^{(1)})}\prod_{i=1}^mg_2(\ld_i,v)+
\pt(v)\prod_{j=1}^ng^{-1}_1(v,\lu_j)\prod_{i=1}^mg_3(v,\ld_i)\, ,
\ee
if the first set of Bethe equations is satisfied.
\be\label{BAE2}
\frac{\pd(\ld_k)}{\pt(\ld_k)}=\prod_{i=1}^ng^{-1}_2(\ld_k,\lu_i)\prod_{\substack{j=1\\j\ne k}}^m\frac{g_3(\ld_k,\ld_j)}{g_2(\ld_j,\ld_k)}\, ,
\ \ \ \  k=1,\cdots, m\, .
\ee
As a consequence of the invariance of the vacuum subspace $\mathcal{H}_0$
under the action of $D(v),$ $B^{(1)}(\ld_1)\cdots
B^{(1)}(\ld_m)|\bar\Omega\rangle$ is a column vector with $2^n$ rows having
vectors in $\mathcal{H}_0\in\mathcal{H}$. Therefore, $B^{(1)}(\ld_1)\cdots
B^{(1)}(\ld_m)|\bar\Omega\rangle$ is an eigenvector of $A(v)$ seen as an
operator on $(\mathbb{C}^2)^{\otimes n}\otimes\mathcal{H}$
\[
A(v)B^{(1)}(\ld_1)\cdots B^{(1)}(\ld_m)|\bar\Omega\rangle=\pu(v)B^{(1)}(\ld_1)\cdots B^{(1)}(\ld_m)|\bar\Omega\rangle\, .
\]
An immediate consequence of this relation and (\ref{ii8}), (\ref{int29}) is that
\[
[\bigotimes_{j=1}^nB(\lu_j)]B^{(1)}(\ld_1)\cdots B^{(1)}(\ld_m)|\bar\Omega\rangle=
B_{i_1}(\lu_1)\cdots B_{i_n}[B^{(1)}(\ld_1)\cdots B^{(1)}(\ld_m)|\bar\Omega\rangle]^{i_1\cdots i_n}\in\mathcal{H}
\]
is an eigenvector of $\Lambda(v)$ with eigenvalue
\be\label{eigenvalue}
\Lambda(v)=\pu(\la)\prod_{j=1}^ng_1(\lu_j,v)+\pd(\la)\prod_{i=1}^ng_2(v,\lu_i)\prod_{j=1}^mg_2(\ld_j,v)
+\pt(v)\prod_{j=1}^mg_3(v,\ld_j)\, ,
\ee
if the second set of Bethe equations
\be\label{BAE1}
\frac{\pu(\lu_l)}{\pd(\lu_l)}=\prod_{\substack{i=1\\i\ne k}}^n\frac{g_2(\lu_l,\lu_i)}{g_1(\lu_i,\lu_l)}\prod_{j=1}^m g_2(\ld_j,\lu_l)\, ,\ \ \
l=1,\cdots,n\, .
\ee
is satisfied.

\subsection{Choosing a different pseudovacuum}\label{ap3}

The formula (\ref{int5}) for the eigenvalues of the QTM differs from
(\ref{eigenvalue}) with parameters (\ref{parametersq}) by a circular
permutation of the chemical potentials $(h_1,h_2,h_3)\rightarrow(h_3,h_1,h_2)$
and grading $(\varepsilon_1,\varepsilon_2,\varepsilon_3)
\rightarrow(\varepsilon_3,\varepsilon_1,\varepsilon_2)$. In this section we
will show how we can obtain (\ref{int5}) by employing a different
pseudovacuum. Making use of (\ref{lq}) and (\ref{llq}) the action of the
monodromy matrix (\ref{mq}) on the pseudovacuum
\be\label{i11}
|\tilde\Omega\rangle=\underbrace{\left(\begin{array}{c} 0\\
0\\1\end{array}\right)\otimes\left(\begin{array}{c} 0\\1\\0
\end{array}\right)\otimes\cdots
\otimes\left(\begin{array}{c} 0\\0\\1 \end{array}\right)\otimes
\left(\begin{array}{c} 0\\1\\0 \end{array}\right)}_{\mbox{ N factors}}\, ,
\ee
is given by
\begin{align}
\T^{QTM}(v)|\tilde\Omega\rangle
&=\left(\begin{array}{ccc} T_{11}^{QTM}(v)|\tilde\Omega\rangle & T_{12}^{QTM}(v)|\tilde\Omega\rangle & T_{13}^{QTM}(v)|\tilde\Omega\rangle \\
                           T_{21}^{QTM}(v)|\tilde\Omega\rangle & T_{22}^{QTM}(v)|\tilde\Omega\rangle & T_{23}^{QTM}(v)|\tilde\Omega\rangle\\
                           T_{31}^{QTM}(v)|\tilde\Omega\rangle & T_{32}^{QTM}(v)|\tilde\Omega\rangle & T_{33}^{QTM}(v)|\tilde\Omega\rangle\\
  \end{array}\right)\\
&=\left(\begin{array}{ccc} e^{\beta h_1}(\beta(v,-u)\beta(u,v))^{N/2}|\tilde\Omega\rangle & T_{12}^{QTM}(v)|\tilde\Omega\rangle &0\\
                              0 & e^{\beta h_2}(\beta(v,-u)\alpha_2(u,v))^{N/2}|\tilde\Omega\rangle & 0\\
                              T_{31}^{QTM}(v)|\tilde\Omega\rangle & T_{32}^{QTM}(v)|\tilde\Omega\rangle  &e^{\beta h_3}(\alpha_3(v,-u)\beta(u,v))^{N/2}|\tilde\Omega\rangle
  \end{array}\right)\, .\nonumber
\end{align}
Consider now the similarity transformation of the Yang-Baxter algebra (\ref{YBA})
\begin{align}
(X\Pi_{12}\Pi_{23}\check\R(v,w)\Pi_{23}^T\Pi_{12}^TX^T)&[X\Pi_{12}\Pi_{23}(\T^{QTM}(v)\otimes\T^{QTM}(w))\Pi_{23}^T\Pi_{12}^TX^T]
=\nonumber
\\&[X\Pi_{12}\Pi_{23}(\T^{QTM}(v)\otimes\T^{QTM}(w))\Pi_{23}^T\Pi_{12}^TX^T](X\Pi_{12}\Pi_{23}\check\R(v,w)\Pi_{23}^T\Pi_{12}^TX^T)\, ,
\end{align}
with $X$ defined in (\ref{defx}) and $\Pi_{12}=\pi_{12}\otimes\pi_{12},$ $\Pi_{23}=\pi_{23}\otimes\pi_{23},$
\be
\pi_{12}=\left(\begin{array}{ccc} 0&1&0\\1&0&0\\0&0&1
  \end{array}\right)\, ,\ \ \ \ \  \ \
\pi_{23}=\left(\begin{array}{ccc} 1&0&0\\0&0&1\\0&1&0
  \end{array}\right)\, .
\ee
The matrix $\pi_{12}$($\pi_{23}$) permutes the states $1$ and $2$ ($2$ and
$3$) in the auxiliary space and
$\Pi_{12}\Pi_{12}^T=\Pi_{12}^T\Pi_{12}=\mathbb{I}_9\,
,\Pi_{23}\Pi_{23}^T=\Pi_{23}^T\Pi_{23}=\mathbb{I}_9$.  Then, we have
\[
\pi_{12}\pi_{23}\T^{QTM}(v)\pi_{23}^T\pi_{12}^T|\tilde\Omega\rangle=
\left(\begin{array}{ccc}   T_{33}^{QTM}(v)|\tilde\Omega\rangle & T_{31}^{QTM}(v)|\tilde\Omega\rangle & T_{32}^{QTM}(v)|\tilde\Omega\rangle \\
                           T_{13}^{QTM}(v)|\tilde\Omega\rangle & T_{11}^{QTM}(v)|\tilde\Omega\rangle & T_{12}^{QTM}(v)|\tilde\Omega\rangle\\
                           T_{23}^{QTM}(v)|\tilde\Omega\rangle & T_{21}^{QTM}(v)|\tilde\Omega\rangle & T_{22}^{QTM}(v)|\tilde\Omega\rangle\\
     \end{array}\right)=
\left(\begin{array}{ccc}   \tilde\pu(v)|\tilde\Omega\rangle& T_{31}^{QTM}(v)|\tilde\Omega\rangle & T_{32}^{QTM}(v)|\tilde\Omega\rangle \\
                            0&\tilde\pd(v)|\tilde\Omega\rangle&T_{12}^{QTM}(v)|\tilde\Omega\rangle\\
                            0&0&\tilde\pt(v)|\tilde\Omega\rangle\\
     \end{array}\right)\, ,
\]
with
\be\label{int30}
\tilde\pu(v)=e^{\beta h_3}(\alpha_3(v,-u)\beta(u,v))^{N/2}\, ,\ \ \tilde\pd(v)=e^{\beta h_1}(\beta(v,-u)\beta(u,v))^{N/2}\, ,\ \ \
\tilde\pt(v)=e^{\beta h_2}(\beta(v,-u)\alpha_2(u,v))^{N/2}
\ee
Together with
\be\label{int31}
X\Pi_{12}\Pi_{23}\check\R(v,w)\Pi_{23}^T\Pi_{12}^TX^T=
\left(\begin{array}{cccc}
         \alpha_3(v,w)&0&0&0\\
         0&\g_-(v,w)\mathbb{I}_2&\beta(v,w)\mathbb{I}_2&0\\
         0&\beta(v,w)\mathbb{I}_2&\g_+(v,w)\mathbb{I}_2&0\\
         0&0&0&\check{\bar{\R}}^{(1)}(v,w)
       \end{array}\right)
\ee
where
\be
\check{\bar{\R}}^{(1)}(v,w)=\left(\begin{array}{cccc}
                          \alpha_1(v,w)&0&0&0\\
                          0&\g_+(v,w)&\beta(v,w)&0\\
                          0&\beta(v,w)&\g_-(v,w)&0\\
                          0&0&0&\alpha_2(v,w)
                        \end{array}\right)\, .
\ee
and
$\Pi_{12}\Pi_{23}(\T^{QTM}(v)\otimes\T^{QTM}(w))\Pi_{23}^T\Pi_{12}^T=[\pi_{12}\pi_{23}\T^{QTM}(v)\pi_{23}^T\pi_{12}^T]
\otimes[\pi_{12}\pi_{23}\T^{QTM}(w)\pi_{23}^T\pi_{12}^T]$ this means that the
ABA for the generalized model applies also for the choice of the pseudovacuum
$|\tilde\Omega\rangle$ obtaining the eigenvalue (\ref{eigenvalue}) and BAEs
(\ref{BAE2}), (\ref{BAE1}) with parameters (\ref{int30}). One could also argue
that in (\ref{int31}) $\g_+(v,w)$ is exchanged with $\g_-(v,w)$ when compared
with the formula for $X\check\R(v,w)X^T$ appearing in (\ref{int19}).  However,
it is easy to see that this exchange does not affect the considerations of
sections \ref{ap1} and \ref{ap2}.  Comparison of (\ref{parametersq}) and
(\ref{int30}) shows that by using the pseudovacuum (\ref{i11}) instead of
(\ref{pseudoQTM}) the only significant change is represented by the circular
permutation of the chemical potentials $(h_1,h_2,h_3)\rightarrow(h_3,h_1,h_2)$
and grading $(\varepsilon_1,\varepsilon_2,\varepsilon_3)
\rightarrow(\varepsilon_3,\varepsilon_1,\varepsilon_2)$.


\begin{thebibliography}{99}
\bibitem{units}   Introducing a length scale $l_0$ via $c=m g l_0/\hbar^2,$ the units of grandcanonical
                  potential per length, temperature, chemical potential and magnetic field are
                  $\phi_0=\hbar^2/(2 m l_0^3 ),$ $T_0=\hbar^2/(2 m l_0^2 k_B)$ and $\mu_0=h_0=
                  \hbar^2/(2 m l_0^2).$
\end{thebibliography}
\end{document}